\documentclass[english,final,table]{IEEEtran}
\usepackage[T1]{fontenc}
\usepackage[latin9]{inputenc}
\usepackage{xcolor}
\usepackage{pdfcolmk}
\usepackage{float}
\usepackage{bm}
\usepackage{amsmath}
\usepackage{amsthm}
\usepackage{amssymb}
\usepackage{graphicx}
\PassOptionsToPackage{normalem}{ulem}
\usepackage{ulem}

\makeatletter

\providecommand{\tabularnewline}{\\}
\floatstyle{ruled}
\newfloat{algorithm}{tbp}{loa}
\providecommand{\algorithmname}{Algorithm}
\floatname{algorithm}{\protect\algorithmname}
\providecolor{lyxadded}{rgb}{0,0,1}
\providecolor{lyxdeleted}{rgb}{1,0,0}

\DeclareRobustCommand{\lyxdeleted}[3]{{\color{lyxdeleted}\lyxsout{#3}}}
\DeclareRobustCommand{\lyxsout}[1]{\ifx\\#1\else\sout{#1}\fi}

\theoremstyle{plain}
\newtheorem{prop}{\protect\propositionname}
\theoremstyle{plain}
\newtheorem{thm}{\protect\theoremname}
\theoremstyle{plain}
\newtheorem{lem}{\protect\lemmaname}

\usepackage{amsmath}
\usepackage{amssymb}
\usepackage[level=0]{wgroup_message}  

\author{
Hao~Sun,~\IEEEmembership{Member,~IEEE},
Junting~Chen,~\IEEEmembership{Member,~IEEE}

\thanks{Manuscript submitted June 17, 2022, revised October 9,2022, November 21, 2022, and accepted December 9, 2022.}

\thanks{This work was supported in part by the National Key R\&D Program of China with grant No. 2018YFB1800800, 
by National Science Foundation of China No. 92067202, No. 62102343, and No. 62171398, 
by Guangdong Research Project No. 2019QN01X895, 
and by the Shenzhen Science and Technology Program No. KQTD20200909114730003 and JCYJ20210324134612033, 
and by Shenzhen Institute of Artificial Intelligence.}
\thanks{H.~Sun is with the Future Network of Intelligence Institute (FNii) and 
School of Science and Engineering, The Chinese University of Hong Kong, Shenzhen,
Guangdong 518172, China (email: haosun1@link.cuhk.edu.cn).}
\thanks{J.~Chen is with the Future Network of Intelligence Institute (FNii) and 
School of Science and Engineering, The Chinese University of Hong Kong, Shenzhen,
Guangdong 518172, China (email: juntingc@cuhk.edu.cn).}}

\usepackage[acronym]{glossaries}
\newcommand{\newac}{\newacronym}
\newcommand{\ac}{\gls}
\newcommand{\Ac}{\Gls}

\makeglossaries

\newac{speb}{SPEB}{square position error bound}
\newac[plural=EFIMs,firstplural=Fisher information matrices (EFIMs)]{efim}{EFIM}{Fisher information matrix}
\newac{ne}{NE}{Nash equilibrium}
\newac{mse}{MSE}{mean squared error}
\newac{toa}{TOA}{time-of-arrival}
\newac{snr}{SNR}{signal-to-noise ratio}
\newac{lan}{LAN}{local area network}
\newac{psd}{PSD}{positive semidefinite}
\newac{pd}{PD}{positive definite}
\newac{wrt}{w.r.t.}{with respect to}
\newac{lhs}{L.H.S.}{left hand side}
\newac{wp1}{w.p.1}{with probability 1}
\newac{kkt}{KKT}{Karush-Kuhn-Tucker}
\newac{wlog}{w.l.o.g.}{without loss of generality}
\newac{mle}{MLE}{maximum likelihood estimation}
\newac{gps}{GPS}{global positioning system}
\newac{rssi}{RSSI}{received signal strength indicator}
\newac{mimo}{MIMO}{multiple-input multiple-output}
\newac{csi}{CSI}{channel state information}
\newac{fdd}{FDD}{frequency division duplexing}
\newac{ms}{MS}{mobile station}
\newac{bs}{BS}{base station}
\newac{d2d}{D2D}{device-to-device}
\newac{slnr}{SLNR}{signal-to-interference-leakage-and-noise-ratio}
\newac{ula}{ULA}{uniform linear antenna array}
\newac{pas}{PAS}{power angular spectrum}
\newac{mmse}{MMSE}{minimum mean square error}
\newac{zf}{ZF}{zero-forcing}
\newac{rzf}{RZF}{regularized zero-forcing}
\newac{as}{AS}{angular spread}
\newac{aod}{AOD}{angle of departure}
\newac{iid}{i.i.d.}{independent and identically distributed} 
\newac{sinr}{SINR}{signal-to-interference-and-noise ratio}
\newac{tdd}{TDD}{time-division duplex}
\newac{rvq}{RVQ}{random vector quantization}
\newac{rhs}{R.H.S.}{right hand side}
\newac{mrc}{MRC}{maximum ratio combining}
\newac{cdf}{CDF}{cumulative distribution function}
\newac{a.s.}{a.s.}{almost surely}
\newac{los}{LOS}{line-of-sight}
\newac{jsdm}{JSDM}{joint spatial division and multiplexing}
\newac{map}{MAP}{maximum a posteriori}
\newac{klt}{KLT}{Karhunen-Lo\`eve Transform}
\newac{lbe}{LBE}{link bargaining equilibrium}
\newac{se}{SE}{Stackelberg equilibrium}
\newac{uav}{UAV}{unmanned aerial vehicle}
\newac{nlos}{NLOS}{non-line-of-sight}
\newac{pdf}{PDF}{probability density function}
\newac{em}{EM}{expectation-maximization}
\newac{knn}{KNN}{$k$-nearest neighbor}
\newac{svd}{SVD}{singular value decomposition}
\newac{nmf}{NMF}{non-negative matrix factorization}
\newac{umf}{UMF}{unimodality-constrained matrix factorization}
\newac{rmse}{RMSE}{rooted mean squared error}
\newac{olos}{OLOS}{obstructed line-of-sight}
\newac{mmw}{mmW}{millimeter wave}
\newac{ber}{BER}{bit error rate}
\newac{rss}{RSS}{received signal strength}
\newac{lp}{LP}{linear program}
\newac{ufw}{U-FW}{unimodal Frank-Wolfe}
\newac{utf}{UTF}{unimodality-constrained tensor factorization}
\newac{fw}{FW}{Frank-Wolfe}
\newac{iot}{IoT}{Internet-of-Things}
\newac{mae}{MAE}{mean absolute error}
\newac{crb}{CRB}{Cram\'er-Rao bound}
\newac{aoa}{AoA}{angle of arrival}
\newac{wcl}{WCL}{weighted centroid localization}

\setkeys{Gin}{width=1.0\columnwidth}

\renewcommand{\lyxdeleted}[3]{{\color{lyxdeleted}{}}}

\newac{als}{ALS}{alternating least square}
\newac{fpca}{FPCA}{fixed point continuation with approximate SVD}
\usepackage[noadjust]{cite}
\usepackage{color}  
\usepackage{graphicx,psfrag,cite,subfigure}

\makeatother

\usepackage{babel}
\providecommand{\lemmaname}{Lemma}
\providecommand{\propositionname}{Proposition}
\providecommand{\theoremname}{Theorem}

\begin{document}
\title{Propagation Map Reconstruction via Interpolation Assisted Matrix Completion}
\maketitle
\begin{abstract}
Constructing a propagation map from a set of scattered measurements
finds important applications in many areas, such as localization,
spectrum monitoring and management. Classical interpolation-type methods
have poor performance in regions with very sparse measurements. Recent
advance in matrix completion has the potential to reconstruct a propagation
map from sparse measurements, but the spatial resolution is limited.
This paper proposes to integrate interpolation with matrix completion
to exploit both the spatial correlation and the potential low rank
structure of the propagation map. The proposed method first enriches
matrix observations using interpolation, and develops the statistics
of the interpolation error based on a local polynomial regression
model. Then, two uncertainty-aware matrix completion algorithms are
developed to exploit the interpolation error statistics. It is numerically
demonstrated that the proposed method outperforms Kriging and other
state-of-the-art schemes, and reduces the \ac{mse} of propagation
map reconstruction by $10$\%\textendash $50$\% for a medium to large
number of measurements.
\end{abstract}


\section{Introduction}

\mysubsubnote{1. What is the big picture? What is the problem? Why is it important? Why is it hard? (1-2 paragraphs)
}Sensing and reconstructing a propagation map find important applications
in many areas. In source localization, constructing the energy propagation
map may help identify the number of sources and localize the source
in harsh environment \cite{ZhuPengLi:J21,CheMit:J17}. Constructing
power spectrum maps, radio maps or coverage maps, provides assistance
in wireless network planning, spectrum allocation, and interference
management \cite{RomKimSeu:J17,MoxHuaXuj:J19,ChoValDra:C16}. In antenna
design, the interpolation and reconstruction of near field electromagnetic
provides useful insights to understand the antenna radiation pattern
which are used in testing and examining the performance and reliability
of wireless devices \cite{FucLeCMig:J17,AusAndNev:J18}. Applications
of propagation map reconstruction are also found in bio-medical studies
\cite{SatWanSun:J22,KeHuaCui:J21} and environment science \cite{CheLiWan:J20,LiuWuDua:J17}.

It is challenging to reconstruct a propagation map from a set of sparse
measurements due to the lack of accurate models to describe the propagation
map. The problem has been studied for more than two decades, but it
is still under active research. Recently developed data-driven approaches
for propagation map reconstruction include Kriging \cite{SatKoyFuj:J17,YakAleTom:J07},
kernel based method \cite{TegRomRam:J19,BazGia:J13,XuZhaZha:C21,HamBer:C17},
dictionary learning \cite{SatWanSun:J22,KimGia:C13}, matrix completion
\cite{BalCohAfa:J22,MigMarDon:J11,SunChe:C21,ZhaFuWan:J20} and deep
learning \cite{WanWanMao:C18,TegYveRom:J21,LevRonYap:J21,ShrFuHong:J22}.
Kriging method predicts the signal strength at a given location by
computing a weighted average of the measured signal strengths in the
neighborhood of the location to be predicted. The weights are obtained
through minimizing the Kriging covariance. However, the performance
of the Kriging method is sensitive with the choice of the parametric
form of the {\em semi-variogram} model in Kriging. Kernel methods
\cite{TegRomRam:J19,BazGia:J13,XuZhaZha:C21,HamBer:C17} aim at transforming
linearly inseparable data to linearly separable one using kernel functions
as a proxy of the propagation map. For example, the \ac{rss} at each
location can be estimated using a weighted average of different radial
basis function \cite{HamBer:C17} where the parameters and weights
are jointly optimized through an alternating minimization method.
Kriging and kernel based methods rely on the assumption that the propagation
map is locally smooth, and therefore, the propagation map can be interpolated
and reconstructed using nearby measurements. It is not surprising
that the reconstruction performance will deteriorate in regions where
measurements are too sparse.

In contrast to interpolation-type methods, matrix completion approaches
try to exploit the global structure of the propagation map. Specifically,
the propagation map is discretized into grid cells and the signal
strength measurements are arranged into a sparse matrix according
to the measurements locations. The matrix is then completed using
compressed sensing algorithms \cite{BalCohAfa:J22,MigMarDon:J11,SunChe:C21,ZhaFuWan:J20}.
These methods are based on the assumption that the matrix representation
of a 2D propagation map has a low rank property. Exploiting such a
global structure of the propagation map, matrix completion methods
improve the reconstruction performance in regions where measurements
are very sparse. However, these methods have poor performance in reconstructing
high resolution propagation map, because the finer the resolution,
the lower percentage of the number of observed entries, and eventually,
the matrix completion algorithm may fail if a complete row or column
of the matrix is not observed.

Recent advance also attempts to develop deep learning techniques for
propagation map reconstruction. In \cite{WanWanMao:C18}, a deep Gaussian
Process is used to model the relationship between \Ac{rss} measurement
values and their corresponding locations. In \cite{TegYveRom:J21},
a fully convolutional deep completion autoencoder architecture is
developed to learn the spatial structure of relevant propagation phenomena,
and a UNet structure is employed in \cite{LevRonYap:J21} for radio
map reconstruction. However, these methods require a large amount
of data to train the neural network model, and the generalization
capability and analytical insights of deep learning approaches are
still not very clear.

\mysubsubnote{3. What is the specific problem to be studied in the paper? (1/2  1 paragraph) -Such specific problem should be relevant to addressing the deficiencies of the state-of-the-art discussed above. Consider to sketch the story by answering these questions: -  What are the fundamental questions relevant to the problem?  - How answering to these questions can impact the design, architecture, and operation? - What is the goal of this paper? (** The goal is usually a high level descriptive objective in one sentence.)}

This paper proposes to {\em integrate} spatial interpolation with
matrix completion for propagation map reconstruction. An interpolation
assisted matrix completion framework is thus developed. Such an integration
is highly non-trivial due to the following two challenges: (1) how
to integrate the two methods such that the integrated approach works
better than both of the methods applied alone; and (2) how to optimize
the parameters for the integrated approach. Specifically, we discretize
the area of interest into grid cells and form a sparse matrix using
a windowing method, where only the grid cells that have sufficient
measurements within a radius of $b$ will be locally constructed via
interpolation. Our preliminary results in \cite{SunChe:C22} revealed
good performance of this strategy under a handpicked parameter $b$.
However, it was not clear how to optimize $b$. The key challenge
is to optimize the window size $b$, which serves as a bridge between
the interpolation and the matrix completion. To achieve such a goal,
we first adopt a local polynomial regression model to estimate the
selected entries of the matrix, and then, we derive analytical results
to characterize the error of the interpolation. The window size $b$
can therefore be optimized to minimize the \ac{mse} of the interpolation.
Finally, we develop two new matrix completion approaches to construct
the propagation map by leveraging the knowledge of uncertainty from
the local interpolation step at the previous stage.

Our experiments found that by combining interpolation with matrix
completion, the accuracy of propagation map reconstruction is substantially
improved. Specifically, the reconstruction \ac{mse} can be reduced
by $10$\%\textendash $50$\% from Kriging and other state-of-the-art
schemes for a medium to large number of measurements, which translates
to a saving of nearly half of the sensor measurements to achieve the
same MSE. We demonstrate the application of the reconstructed propagation
map to \Ac{rss}-based source localization, where the root mean squared
error (RMSE) of localization is reduced by more than $50$\% from
a \ac{wcl} \cite{WanUrrHanCab:J11} baseline.

\mysubsubnote{5. What is the contribution of the paper? (1/2  1 paragraph) 	- What is the high-level contribution? (1 sentence] 	- What are we advocating?  	- What are the interesting findings?  	- What are the key contributions? (a few bullet points) For example: We develop a framework for ...; We establish foundation for ...; We determine the fundamental limits of ...; We introduce the concepts of ...; We put forth the notion of ...; We propose a new model ...;  We characterize ...; We design ...; We analyze...; We derive ...; We prove that ...; We show that ...; We quantify the ...}

To summarize, the following contributions are made:
\begin{itemize}
\item We develop a windowing-based integrated interpolation and matrix completion
framework for propagation map reconstruction, where the window size
balances the contribution between the interpolation and the matrix
completion.
\item Based on two local polynomial regression models, we analyze the moments
and the asymptotic distribution of the interpolation error. In addition,
we develop a minimum \ac{mse} approach for adjusting the window size
to optimize the integrated interpolation and matrix completion framework.
\item We develop two uncertainty-aware matrix completion methods to integrate
the construction from the local interpolation. Our numerical results
reveal that the proposed integrated approaches beat the conventional
ones with substantial improvement in the accuracy of propagation map
reconstruction.
\end{itemize}

The rest of the paper is organized as follows. Section \ref{sec:System-model}
establishes the propagation reconstruction model. Section \ref{sec:Local-Polynomial-Regression}
develops a local polynomial regression method to construct a sparse
matrix, analyzes the interpolation error and proposes the window size
selection approach. Section \ref{sec:global reconstruction} numerically
verifies the low rank property and proposes the uncertain-aware matrix
completion scheme. Numerical results are presented in Section \ref{sec:Numerical-Results}
and conclusion is given in Section \ref{sec:Conclusion}.

\emph{Notation:} Vectors are written as bold italic letters $\bm{x}$
and matrices as bold capital italic letters $\bm{X}$. For a matrix
$\bm{X}$, $X_{ij}$ denotes the entry in the $i$th row and $j$th
column of $\bm{X}$. The notation $\bm{c}_{ij}$ denotes the center
location of the $(i,j)$th grid, $o(x)$ means $\text{lim}_{x\to0}o(x)/x\rightarrow0$,
$O(x)$ means $|O(x)|/x\leq C$, for all $x>x_{0}$ with $C$ and
$x_{0}$ are positive real numbers, $\text{diag}(\bm{X})$ represents
a column vector whose entries are the diagonal elements of matrix
$\bm{X}$, and $\text{diag}(\bm{x})$ represents a diagonal matrix
whose diagonal elements are the entries of vector $\bm{x}$. Symbol
$\mathbb{E}\left\{ \cdot\right\} $ and $\mathbb{V}\left\{ \cdot\right\} $
denote expectation and variance separately.

\section{System model\label{sec:System-model}}

\subsection{Reconstruction Model}

\mysubsubnote{- Propagation field model due to one (or more) source(s). (PS: perhaps we need to handle the case where the propagation field consists of 2-3 sources.)}

Consider a propagation field that is excited by $S$ sources located
at $\bm{s}_{k}\in\mathcal{D}$, $k=1,2,\dots,S$, in an area $\mathcal{D}\subset\mathbb{R}^{2}$.
The signal emitted from the sources is detected by $M$ sensors with
known locations $\bm{z}_{m}\in\mathbb{R}^{2}$, $m=1,2,\dots,M$,
randomly deployed in $\mathcal{D}$. The propagation map to be reconstructed
is modeled as 
\begin{equation}
\rho(\bm{z})\triangleq\sum_{k=1}^{S}g_{k}(d(\bm{s}_{k},\bm{z}))+\zeta(\bm{z})\qquad\bm{z}\in\mathcal{D}\label{eq:model-propagation-field}
\end{equation}
where $d(\bm{s},\bm{z})=\|\bm{s}-\bm{z}\|_{2}$ describes the distance
between a source at $\bm{s}$ and a sensor at $\bm{z}$, $g_{k}(d)$
describes the propagation function from the $k$th source in terms
of the propagation distance $d$, and the term $\zeta(\bm{z})$ is
a random component that captures the shadowing which is assumed to
have spatial correlations.

The strength of the signal received by the $m$th sensor is given
by 
\begin{equation}
\gamma_{m}=\rho(\bm{z}_{m})+\epsilon_{m}\label{eq:model-measurement}
\end{equation}
where $\epsilon_{m}$ is a random variable with zero mean and variance
$\sigma^{2}$ to model the measurement noise. The goal of this paper
is to reconstruct $\rho(\bm{z})$ based on $M$ RSS measurements $\{(\bm{z}_{m},\gamma_{m})\}$.

\mysubsubnote{- Field reconstruction model and elevation metric, i.e., state that the goal is to reconstruct a propagation field with N by N resolution (N is given, and thus, we don't need to discuss the choice of N). Thus, we need to bring out the discretizaiton, matrix formation, etc.}

For a given spatial resolution $N\times N$ for the propagation map
reconstruction, we consider to discretize the target area $\mathcal{D}$
into $N$ rows and $N$ columns that results in $N^{2}$ grid cells.
Let $\bm{c}_{ij}\in\mathcal{D}$ be the center location of the $(i,j)$th
grid cell, and $\bm{H}$ be a matrix representation of the propagation
map $\rho(\bm{z})$, where the $(i,j)$th entry is defined as $H_{ij}=\rho(\bm{c}_{ij})$.
As a result, the objective is to estimate a matrix $\bm{\bar{H}}$
based on the $M$ RSS measurements, such that the squared error $||\bm{\bar{H}}-\bm{H}||_{F}^{2}=\sum_{i,j}(\bar{H}_{ij}-\rho(\bm{c}_{ij}))^{2}$
is as low as possible.\mysubsubnote{make this consistent with your numerical results: whether you present the performance in terms of the total squared error or mean squared error.}

We shall highlight that the parametric forms of the propagation models
$g_{k}(d)$ are \emph{unknown}, except that $g_{k}(d)$ are believed
to be smooth and decrease in distance $d$. Thus, model (\ref{eq:model-propagation-field})
tends to have a low rank structure. The intuition is that the path
gain is usually dominated by the propagation distance $d(\bm{s},\bm{z})$
between the source location $\bm{s}$ and the measurement location
$\bm{z}$, and hence, there is a hidden homogeneity in all directions
from $\bm{s}$. In addition, the statistics of the shadowing $\zeta(\bm{z})$
is also unknown. Therefore, classical parametric methods fail to apply
here. On the other hand, conventional interpolation-based approaches
fail to exploit the global structure of $g_{k}(d)$. For example,
$g_{k}(d(\bm{s}_{k},\bm{z}_{1}))$ may equal to $g_{k}(d(\bm{s}_{k},\bm{z}_{2}))$
despite $\bm{z}_{1}$ and $\bm{z}_{2}$ being possibly far apart if
the distances are equal, \emph{i.e.}, $d(\bm{s}_{k},\bm{z}_{1})=d(\bm{s}_{k},\bm{z}_{2})$.
Such a property implies that local interpolation approaches are strictly
sub-optimal, and one should connect the entire set of measurements
to improve propagation map reconstruction at a global scale.

\subsection{Interpolation Assisted Matrix Completion\label{subsec:Field-Reconstruction-via}}

Denote $\bm{H}^{(k)}\in\mathbb{R}^{N\times N}$ as the propagation
matrix associated with the $k$th source, with the $(i,j)$th entry
given by $H_{i,j}^{(k)}=g_{k}(d(\bm{s}_{k},\bm{c}_{ij}))$. Then,
$\bm{H}=\sum_{k=1}^{S}\bm{H}^{(k)}+\bm{\zeta}.$ As will be numerically
evaluated in Section \ref{subsec:lowrank simulation}, $\sum_{k=1}^{S}\bm{H}^{(k)}$
is likely to be low rank, and so is $\sum_{k=1}^{S}\bm{H}^{(k)}+\bm{\zeta}$
due to the spatial correlation. As a result, propagation map reconstruction
naturally leads to a sparse matrix completion problem.

\subsubsection{Constructing a Sparse Observation Matrix $\hat{\bm{H}}$}

A straight-forward approach to form a sparse matrix is to assign the
measurement $\gamma_{m}$ to $\hat{H}_{ij}$ if the $m$th sensor
$\bm{z}_{m}$ locates in the $(i,j)$th grid centered at $\bm{c}_{ij}$
and is the closest one to $\bm{c}_{ij}$. However, the number of measurements
$M$ could be substantially smaller than the number of grid cells
$N^{2}$ under high resolution reconstruction, resulting in an overly
sparse matrix that is difficult to complete. Moreover, there could
be significant discretization error since the sensor location $\bm{z}_{m}$
may be away from the grid center $\bm{c}_{ij}$.

We propose to interpolate a subset of grid cells $\Omega$ to enrich
the observations for matrix completion. There are two approaches.
(1) Uniform sampling: we form the observation set $\Omega$ by sampling
$CN\text{log}^{2}(N)$ grid cells uniformly at random, where the parameter
$C$ can be empirically chosen to guarantee a sufficient number of
observations for ensuring the identifiability of the matrix completion
\cite{CanPla:J10}. (2) Sensor-aware sampling: we form an observation
set of grid cells where there are sufficient measurements nearby.
Specifically, given a window size parameter $b$ and a measurement
number $M_{0}$, we define an observation set $\Omega$ as a subset
of grids $(i,j)$ such that there are at least $M_{0}$ sensors locating
within a radius of $b$ from the grid center $\bm{c}_{ij}$, \emph{i.e.},
\begin{equation}
\Omega=\Big\{(i,j):\sum_{m=1}^{M}\mathbb{I}\{\|\bm{z}_{m}-\bm{c}_{ij}\|_{2}<b\}\geq M_{0}\Big\}\label{eq:model-sparse-observation}
\end{equation}
where $\mathbb{I}\{A\}=1$ if condition $A$ is satisfied, and $\mathbb{I}\{A\}=0$
otherwise.

As a result, one only estimates $\hat{H}_{ij}$ for $(i,j)\in\Omega$
using the nearby measurements as illustrated in Fig.~\ref{fig:(a)-Propagation-map}.
The identifiability issue for matrix completion under the sensor-aware
sampling will be discussed in Section \ref{subsec:Identifiability-of-the}.

\subsubsection{Formulating the Matrix Completion Problem\label{subsec:Formulating-the-Matrix}}

One representative approach is the \ac{als} algorithm \cite{TanNeh:J11,DavMar:J12},
which minimizes $||\bm{y}-\mathcal{A}(\bm{X})||_{2}$, subject to
a product model $\bm{X}=\bm{LR}$ for $\bm{L}\in\mathbb{R}^{N\times p}$
and $\bm{R}\in\mathbb{R}^{p\times N}$, where $\bm{y}\in\mathbb{R}^{M}$
is the measurement vector and $\mathcal{A}:\mathbb{R}^{N\times N}\to\mathbb{R}^{M}$
is a sensing operator that maps the $(i,j)$th element of the matrix
$\bm{X}$ to the $m$th element of a vector to be compared with $\bm{y}$
\cite{FouRau:13}. Another possibility is to minimize the nuclear
norm $\|\bm{X}\|_{*}$ of a matrix $\bm{X}$ subject to observation
noise $|X_{ij}-\hat{H}_{ij}|\leq\epsilon$. It will be shown later
that these existing approaches do not perform well, since they fail
to exploit the fact that the observed entries may have different uncertainty
according to the methods that form the matrix observations $\hat{H}_{ij}$.
\begin{figure}
\subfigure{\includegraphics[width=0.5\columnwidth]{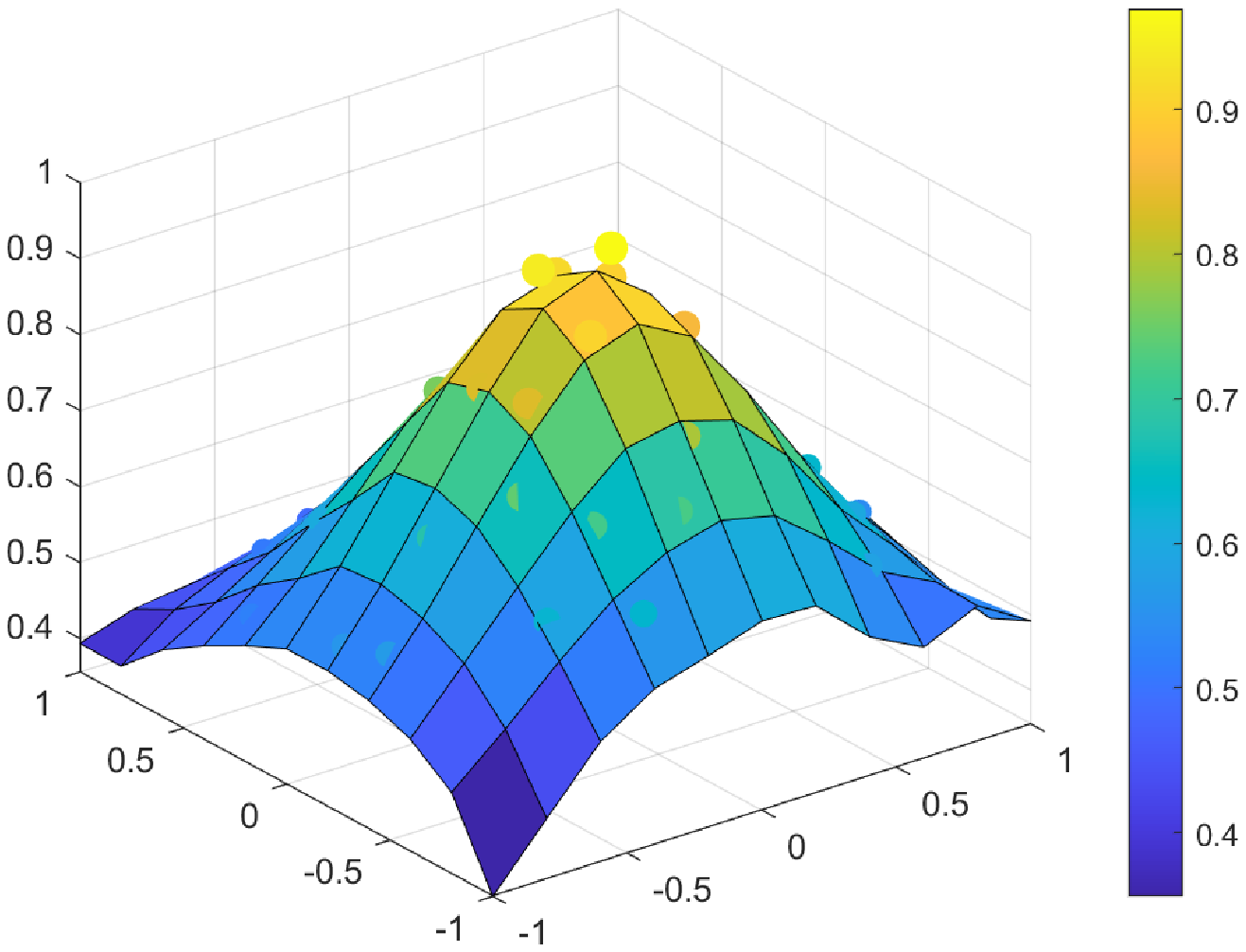}}\subfigure{\includegraphics[width=0.5\columnwidth]{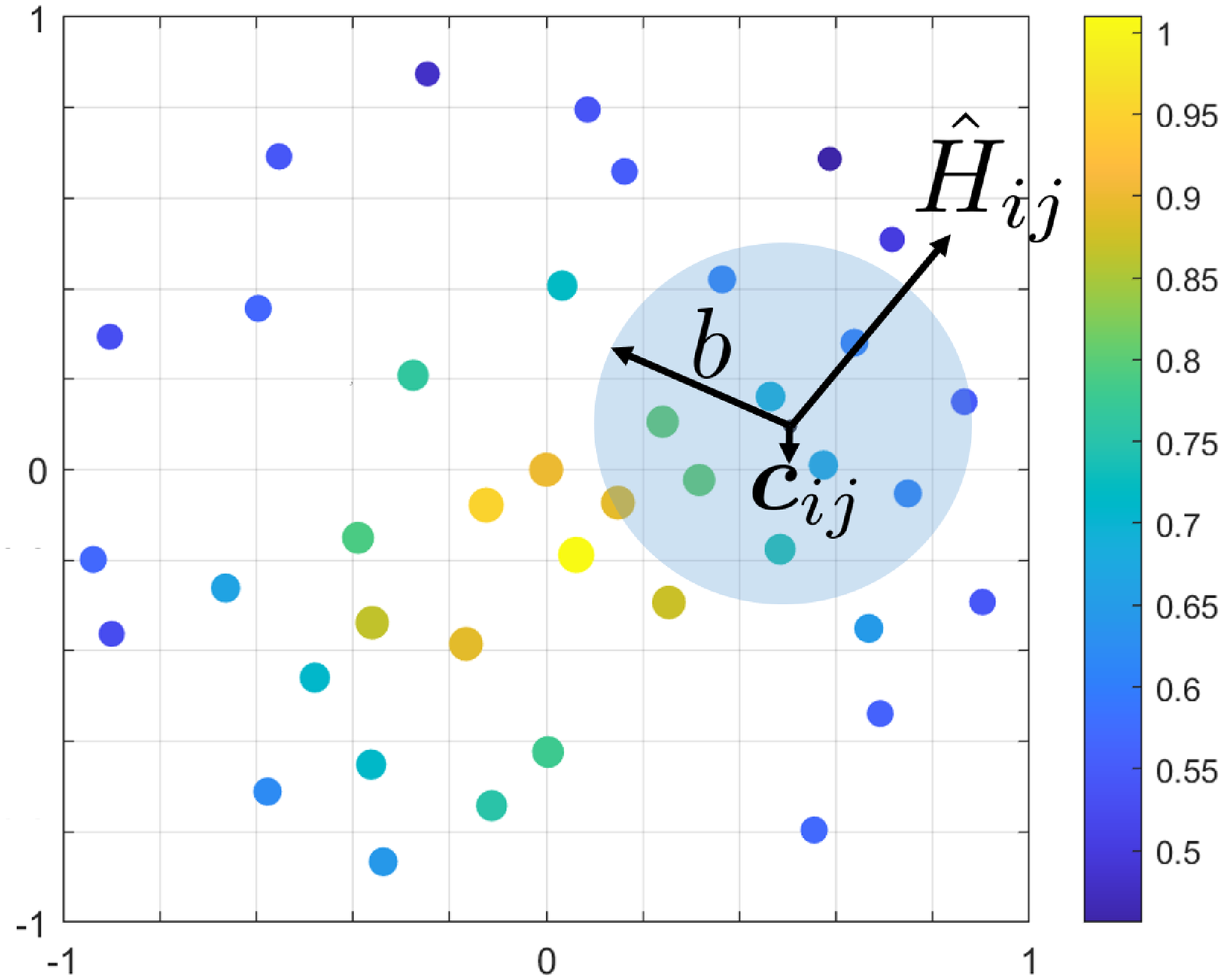}}

\qquad  \ \ \,\, \,\,\,\,\,\, \,\quad \footnotesize (a) \qquad \qquad \qquad \qquad \qquad \qquad \qquad  \, \footnotesize(b)

\caption{\label{fig:(a)-Propagation-map}(a) Propagation map reconstruction
based on local polynomial regression (b) Sensor measurements (colored
dots) within a range of $b$ (grey region) from the grid center $\bm{c}_{ij}$
are used to estimate $\hat{H}_{ij}$. \mysubsubnote{I would prefer to illustrate the idea of local polynomial regression using a 1-D example. The 2-D example looks complicated.}}
\end{figure}

\section{Local Reconstruction via \\Polynomial Regression\label{sec:Local-Polynomial-Regression}}

It is crucial to determine the window size $b$ in the proposed integrated
approach. Intuitively, if the observation set $\Omega$ is constructed
by sampling uniformly at random, the reconstruction performance mainly
depends on the quality of the interpolation for the entries in $\Omega$.
The interpolation performance depends on the parameter $b$ as will
be shown later. If $\Omega$ is constructed by a sensor-aware approach,
then for a small $b$, matrix completion dominates the performance,
since many of entries are missing; and for a large $b$, interpolation
dominates, since the local interpolation fully constructs the matrix.

In this section, we exploit local polynomial regression\footnote{Note that, in this paper, although we derive the results based on
a local polynomial regression model for the local interpolation, the
same methodology can be extended to the integration with other interpolation
methods, such as Kriging.} to construct matrix observations $\hat{H}_{ij}$. Then, we derive
analytical results to characterize the construction error in terms
of the parameter $b$, and optimize for $b$.

\subsection{Interpolation based on Local Polynomial Regression}

Consider to employ a \emph{parametric model} $\hat{\rho}(\bm{z};\bm{c})$
to locally approximate the propagation map $\rho(\bm{z})$ in the
neighborhood of $\bm{c}$. When the area of interest around $\bm{c}$
is small enough compared with the variation of the propagation map
$\rho(\bm{z})$, one may employ a zeroth order model 
\begin{equation}
\hat{\rho}(\bm{z};\bm{c})=\alpha(\bm{c})\label{eq:zero-th}
\end{equation}
or a first order model 
\begin{equation}
\hat{\rho}(\bm{z};\bm{c})=\alpha(\bm{c})+\bm{\beta}^{\text{T}}(\bm{c})(\bm{z}-\bm{c})\label{eq:first order}
\end{equation}
to approximate $\rho(\bm{z})$. The coefficients $\alpha(\bm{c})$
and $\bm{\beta}(\bm{c}\text{)}$ depend on the location $\bm{c}$,
and can be estimated using the measurements $\{(\bm{z}_{m},\gamma_{m})\}$
obtained near $\bm{c}$. The extension to a higher order model is
straight-forward.

Specifically, let $\bm{\theta}=\{\alpha(\bm{c}),\bm{\beta}(\bm{c}),...\}$
be the set of coefficients for the polynomial $\hat{\rho}(\bm{z};\bm{c},\bm{\theta})$.
Then, $\bm{\theta}$ can be estimated using distance-weighted least-squares
regression: 
\begin{equation}
\text{\ensuremath{\underset{\bm{\theta}}{\text{minimize}}}}\ \sum_{m=1}^{M}\big(\gamma_{m}-\hat{\rho}(\bm{z}_{m};\bm{c},\bm{\theta})\big)^{2}K\left(\frac{\bm{z}_{m}-\bm{c}}{b}\right)\label{eq:LS}
\end{equation}
where $b$ is the window parameter as introduced in (\ref{eq:model-sparse-observation})
and $K(\bm{u})$ is a two dimensional kernel function that satisfies
the following conditions:
\begin{itemize}
\item[(i)] $K(\bm{u})$ is a non-negative, symmetric, and bounded \ac{pdf}
satisfying 
\begin{equation}
\int K(\bm{u})d\bm{u}=\iint K(\bm{u})du_{x}du_{y}=1\label{eq:integral 1}
\end{equation}
\begin{equation}
\iint K(\bm{u})^{p_{0}}u_{x}^{p_{1}}u_{y}^{p_{2}}du_{x}du_{y}=0\label{eq:integral 0}
\end{equation}
for $p_{0}\in\{1,2\}$ and $p_{1}$, $p_{2}$ $\in\{1,3\}$.
\item[(ii)] $K(\bm{u})$ has a compact support 
\begin{equation}
\mathcal{C}=\left\{ \text{\ensuremath{\bm{u}}}\in\mathbb{R}^{2}:||\bm{u}||_{2}<1\right\} \label{eq:support}
\end{equation}
 where $K(\bm{u})=0$ for $\bm{u}\notin\mathcal{C}$.
\end{itemize}
Thus, the regression problem in (\ref{eq:LS}) only considers the
measurements that in a radius $b$ from $\bm{c}$ and assigns a high
weight if the measurement location $\bm{z}_{m}$ is close to $\bm{c}$,
and a low weight if $\bm{z}_{m}$ is far away from $\bm{c}$. For
instance, the Epanechnikov kernel \cite{Fan:b96} $K(\bm{u})=\max\{0,\frac{3}{4}(1-||\bm{u}||^{2})\}$
and the truncated Gaussian kernel $K(\bm{u})=\frac{50}{17\sqrt{2\pi}}\text{exp}(-||\bm{u}||^{2}/2)\text{\ensuremath{\mathbb{I}}}\{||\bm{u}||^{2}\leq1\}$
satisfy the above conditions.

As a result, for a given location $\bm{c}$, the local parameter $\hat{\alpha}(\bm{c})$
under a zeroth order model is found as the solution to 
\begin{equation}
\text{\ensuremath{\underset{\alpha}{\text{minimize}}}}\text{\ }\sum_{m=1}^{M}\big(\gamma_{m}-\alpha(\bm{c})\big)^{2}K\left(\frac{\bm{z}_{m}-\bm{c}}{b}\right)\label{eq:LS-0-order}
\end{equation}
whereas, the local parameter $(\hat{\alpha}(\bm{c}),\hat{\bm{\beta}}(\bm{c}))$
under a first order model is found as the solution to
\begin{equation}
\text{\ensuremath{\underset{\alpha,\bm{\beta}}{\text{minimize}}}\ }\sum_{m=1}^{M}\big(\gamma_{m}-\alpha(\bm{c})-\bm{\beta}^{\text{T}}(\bm{c})(\bm{z}_{m}-\bm{c})\big)^{2}K\left(\frac{\bm{z}_{m}-\bm{c}}{b}\right).\label{eq:LS-1-order}
\end{equation}

There exists closed-form solutions.
\begin{prop}[Interpolation]
 Given a location $\bm{c}$, the solution to (\ref{eq:LS-0-order})
under the zeroth order model is given by

\begin{equation}
\hat{\alpha}(\bm{c})=\left(\sum_{m=1}^{M}w_{m}(\bm{c})\gamma_{m}\right)\Big/\sum_{m=1}^{M}w_{m}(\bm{c})\label{eq:sol-0th}
\end{equation}
where $w_{m}(\bm{c})=K((\bm{z}_{m}-\bm{c})/b)$ denotes the kernel
weight for the $m$th measurement.

The solution to (\ref{eq:LS-1-order}) under the first order model
is given by\textup{ 
\begin{equation}
\left[\begin{array}{c}
\hat{\alpha}(\bm{c})\\
\hat{\bm{\beta}}(\bm{c})
\end{array}\right]=(\tilde{\bm{D}}\bm{W}\tilde{\bm{D}}^{\text{T}})^{-1}\tilde{\bm{D}}\bm{W}\bm{\gamma}\label{eq:firstsol}
\end{equation}
where $\bm{\gamma}\triangleq(\gamma_{1},\gamma_{2},\cdots,\gamma_{M})^{\text{T}}$}
is a vector form of the measurements, $\bm{W}=\mbox{diag}\{(w_{1}(\bm{c}),w_{2}(\bm{c}),\cdots,w_{M}(\bm{c}))\}$,
and
\[
\tilde{\bm{D}}=\left[\begin{array}{c}
\bm{1}^{\text{T}}\\
\bm{D}
\end{array}\right]
\]
in which $\bm{1}$ is a vector of all $1$\textquoteright s and $\bm{D}$
is a $2\times M$ matrix with the $m$th column given by\textup{ }$\bm{z}_{m}-\bm{c}$.
\end{prop}
\mysubsubnote{please change $Z_{1}$ into $\tilde{D}$} \begin{proof}
Note that the problems (\ref{eq:LS-0-order}) and (\ref{eq:LS-1-order})
are unconstrained quadratic optimization problems and can be easily
verified to be convex. Thus, the solution can be obtained by re-arranging
the objective into a matrix-vector form and setting the derivative
to zero \cite{Fan:b96,boyd2004convex}. \end{proof} As a result,
to construct the $(i,j)$th observation $\hat{H}_{ij}$ at the grid
centered at $\bm{c}_{ij}$, one can first construct a local model
$\hat{\rho}(\bm{z};\bm{c}_{ij})$ at location $\bm{c}_{ij}$. Then,
$\hat{H}_{ij}$ can be estimated from $\hat{\rho}(\bm{z}=\bm{c}_{ij};\bm{c}_{ij})=\alpha(\bm{c}_{ij})$
for both zeroth order and first order models. From the solution (\ref{eq:LS-0-order})
for the zeroth order model, $\hat{H}_{ij}=\left(\sum_{m}w_{m}(\bm{c}_{ij})\gamma_{m}\right)/\sum_{m}w_{m}(\bm{c}_{ij})$
is simply the distance-weighted strength $\gamma_{m}$ from the measurements
around the grid center $\bm{c}_{ij}$. For the first order model,
$\hat{H}_{ij}=\hat{\alpha}(\bm{c}_{ij})$ as given in (\ref{eq:firstsol})
with setting $\bm{c}=\bm{c}_{ij}$.

\subsection{The Error of Interpolation}

\mysubsubnote{Analyse the bias and variance under finite M for 0th
order model and 1st order model respectively}

An unsolved yet important issue is the choice of the window parameter
$b$ in the estimation problem (\ref{eq:LS}). Intuitively, the zeroth
order model is easy to estimate as it has just one parameter, but
it has poor inference capability away from $\bm{c}$, and as a result,
it may only work well in a small local area under a small window size
$b$ where the propagation map $\rho(\bm{z})$ varies slowly. On the
other hand, a higher order model may have less bias to infer $\rho(\bm{z})$
away from $\bm{c}$, but it requires a lot more measurement data as
it has more parameters, and thus, it prefers a large $b$ to include
more measurements.

Recall the interpolation error $\xi_{ij}\triangleq\hat{H}_{ij}-\rho(\bm{c}_{ij})$,
and $\hat{H}_{ij}=\hat{\alpha}(\bm{c}_{ij})$ which can be computed
using (\ref{eq:sol-0th}) or (\ref{eq:firstsol}) given the center
location $\bm{c}_{ij}$ of the $(i,j)$th grid. In the following,
unless specifically pointed out, we take all the expectation and variance
over the random measurement noise $\epsilon_{m}$ given the sensor
locations $\bm{z}_{m}$. Assume that $\rho(\text{\ensuremath{\bm{z}}})$
is second order differentiable. For a fixed sensor deployment $\bm{z}_{1},\bm{z}_{2},\dots,\bm{z}_{M}$,
the bias $\mathbb{E}\{\xi_{ij}\}$ and variance $\mathbb{V}\{\xi_{ij}\}$
of the interpolation error for $\hat{H}_{ij}$ can be derived and
summarized in the following theorem.
\begin{thm}[Error under zeroth order model]
\label{(estimation-error)} For the zeroth order model, 
\begin{align}
\mathbb{E}\{\xi_{ij}\} & =\sum_{m=1}^{M}\nabla\rho(\bm{c}_{ij})^{\text{T}}\left(\bm{z}_{m}-\bm{c}_{ij}\right)\bar{w}_{m}(\bm{c}_{ij})+o(b)\label{eq:0th-bias}\\
\mathbb{V}\{\xi_{ij}\} & =\sum_{m=1}^{M}\bar{w}_{m}^{2}(\bm{c}_{ij})\sigma^{2}\label{eq:0th-variance}
\end{align}
\label{thm1:0-th bv}\textup{where} $\bar{w}_{m}(\bm{c}_{ij})\triangleq w_{m}(\bm{c}_{ij})/\sum_{i=1}^{M}w_{i}(\bm{c}_{ij})$
\textup{is the normalized weight for the $m$th measurement, and $\nabla{\rho}(\bm{c}_{ij})$
is the derivative of the propagation map $\rho(\bm{z})$ at location
$\bm{z}=\bm{c}_{ij}$. The term $o(b)$ represents a residual and
it satisfies $o(b)/b\to0$ as $b\to0$.}
\end{thm}
\begin{proof} See Appendix \ref{app:theorem1}. \end{proof} For
the result under the first order model, denote $\bm{D}_{ij}$ as a
$2\times M$ matrix to capture the direction from the grid center
$\bm{c}_{ij}$ to the sensor location $\bm{z}_{m}$, and the $m$th
column of $\bm{D}_{ij}$ is defined as $\bm{z}_{m}-\bm{c}_{ij}$.
In addition, define a direction matrix with offsets as $\tilde{\bm{D}}_{ij}=\left[\begin{array}{cc}
\bm{1} & \bm{D}_{ij}^{\text{T}}\end{array}\right]^{\text{T}}.$
\begin{thm}[Error under first order model]
\label{thm2:1st bv}For the first order model, 
\begin{align}
\mathbb{E}\{\xi_{ij}\} & =\frac{1}{2}\left[\tilde{\bm{W}}_{ij}^{-1}\tilde{\bm{D}}_{ij}\bm{W}_{ij}\text{diag}\{\bm{D}_{ij}^{\text{T}}\bm{\Psi}_{ij}\bm{D}_{ij}\}\right]_{(1,1)}+o(b^{2})\label{eq:1st bias}\\
\mathbb{V}\{\xi_{ij}\} & =\sigma^{2}\left[\tilde{\bm{W}}_{ij}^{-1}\left(\tilde{\bm{D}}_{ij}\bm{W}_{ij}\bm{W}_{ij}^{\text{T}}\tilde{\bm{D}}_{ij}^{\text{T}}\right)\tilde{\bm{W}}_{ij}^{-1}\right]_{(1,1)}\label{eq:1st-variance}
\end{align}
where $\tilde{\bm{W}}_{ij}=\tilde{\bm{D}}_{ij}\bm{W}_{ij}\tilde{\bm{D}}_{ij}^{\text{T}}$,
the operation $[\bm{A}]_{(1,1)}$ returns the $(1,1)$th entry of
a matrix $\bm{A}$, and $\bm{\Psi}_{ij}=\nabla^{2}\rho(\bm{c}_{ij})$
is the Hessian matrix of $\rho(\bm{z})$ at point $\bm{c}_{ij}$.
\textup{The term $o(b^{2})$ represents a residual and it satisfies
$o(b^{2})/b^{2}\to0$ as $b\to0$.}
\end{thm}
\begin{proof} See Appendix \ref{app:the2}. \end{proof} \mysubsubnote{Please
study the text below. We need this depth/detail of discussion. }

As seen from Theorems \ref{thm1:0-th bv} and \ref{thm2:1st bv},
the bias $\mathbb{E}\{\xi_{ij}\}$ depends on the variation of the
propagation map $\rho(\bm{z})$. Specifically, under the zeroth order
model, the slope $\nabla\rho(\bm{c}_{ij})$ of the propagation map
contributes to the bias. It follows that, for $\bm{c}_{ij}$ far away
from any source, the bias $\mathbb{E}\{\xi_{ij}\}$ tends to be small
because both $\rho(\bm{c}_{ij})$ and $\nabla\rho(\bm{c}_{ij})$ are
small for typical propagation maps. On the other hand, when $\bm{c}_{ij}$
locates in the area where the slope $\nabla\rho(\bm{c}_{ij})$ is
large, the bias $\mathbb{E}\{\xi_{ij}\}$ will be large, and thus,
the zeroth order interpolation may perform poorly. Under the first
order model, the bias $\mathbb{E}\{\xi_{ij}\}$ is not affected by
the slope but by the curvature of the propagation map, \emph{i.e.},
$\nabla^{2}\rho(\bm{c}_{ij})$.

A bias-variance trade-off due to the window size parameter $b$ can
be observed. First, it is not surprising to find that the absolute
bias $|\mathbb{E}\{\xi_{ij}\}|$ increases for large window size $b$,
if one numerically evaluates the expressions in (\ref{eq:0th-bias})
and (\ref{eq:1st bias}). This is because the estimator $\hat{H}_{ij}$
tends to fit to the structure of a larger area, but the local parametric
model $\hat{\rho}(\bm{z};\bm{c})$ is accurate only for a small area
around $\bm{c}$. Second, by contrast, the variance decreases as $b$
increases, which can be verified by numerically evaluating the variance
expressions in (\ref{eq:0th-variance}) and (\ref{eq:1st-variance}).
The reason is that a large window size $b$ can include more sensor
measurements, and thus, the measurement noise in (\ref{eq:model-measurement})
can be suppressed.

\subsection{Asymptotic Analysis of the Interpolation Error}

\label{subsec:Asymptotic-Results}

\mysubsubnote{Extend the results to asymptotic regime, and draw
some insights.}

To explicitly analyze the bias and variance trade off, we derive the
asymptotic distribution of the interpolation error $\xi_{ij}$ under
the regime of large number of sensors.

Define $f(\bm{z})$ as the density function of the sensors deployed
at location $\bm{z}$. Assume that $f(\bm{z})$ is second order differentiable
and the sensors are \ac{iid} according to $f(\bm{z})$. Moreover,
assume that the propagation function $\rho(\bm{z})$ is third order
differentiable. We have the following results on the asymptotic distribution
of the interpolation error $\xi_{ij}$.
\begin{thm}[Asymptotic interpolation error I]
\label{thm:asymp 0th}For a small enough $b$, the interpolation
error $\xi_{ij}$ under the zeroth order interpolation (\ref{eq:sol-0th})
converges as 
\begin{equation}
\xi_{ij}\overset{p}{\to}b^{2}C_{0}\Big(\frac{\vartheta_{1}(\bm{c}_{ij})}{f(\bm{c}_{ij})}+\frac{1}{2}\vartheta_{2}(\bm{c}_{ij})\Big)+o(b^{2})\label{eq:asym 0 mean}
\end{equation}
as $M\to\infty$, where $\overset{p}{\to}$ means convergence in probability,
and the centered error $\bar{\xi}_{ij}=\xi_{ij}-\mathbb{E}\{\xi_{ij}\}$
converges as 
\begin{equation}
\sqrt{Mb^{2}}\bar{\xi}_{ij}\overset{d}{\to}\mathcal{N}\Big(0,\frac{C_{1}\sigma^{2}}{f(\bm{c}_{ij})}+o(b)\Big)\label{eq:asym 0 distri}
\end{equation}
as $M\to\infty$, where $\overset{d}{\to}$ means convergence in distribution,
\[
\vartheta_{1}(\bm{c}_{ij})=\frac{\partial f(\bm{c}_{ij})}{\partial u_{x}}\frac{\partial\rho(\bm{c}_{ij})}{\partial u_{x}}+\frac{\partial f(\bm{c}_{ij})}{\partial u_{y}}\frac{\partial\rho(\bm{c}_{ij})}{\partial u_{y}}
\]
 
\[
\vartheta_{2}(\bm{c}_{ij})=\frac{\partial^{2}\rho(\bm{c}_{ij})}{\partial u_{x}^{2}}+\frac{\partial^{2}\rho(\bm{c}_{ij})}{\partial u_{y}^{2}}
\]
\textup{$C_{0}=\int u_{x}^{2}K(\bm{u})du_{x}$,} and\textup{ $C_{1}=\int K(\bm{u})^{2}d\bm{u}.$}
\end{thm}
\begin{proof} See Appendix \ref{app:the3}. \end{proof}

Following a similar technique as illustrated in Appendix \ref{app:the3},
we obtain the following result on the interpolation error $\xi_{ij}$
under the first order interpolation (\ref{eq:firstsol}).
\begin{thm}[Asymptotic interpolation error II]
\label{thm:asymp 1st }For a small enough $b$, the interpolation
error $\xi_{ij}$ under the first order interpolation (\ref{eq:firstsol})
converges as 
\[
\xi_{ij}\overset{p}{\to}\frac{1}{2}b^{2}C_{0}\vartheta_{2}(\bm{c}_{ij})+o(b^{2})
\]
as $M\to\infty,$ and the centered error $\bar{\xi}_{ij}=\xi_{ij}-\mathbb{E}\{\xi_{ij}\}$
converges as 
\[
\sqrt{Mb^{2}}\bar{\xi}_{ij}\overset{d}{\to}\mathcal{N}\Big(0,\frac{C_{1}\sigma^{2}}{f(\bm{c}_{ij})}+o(b)\Big)
\]
as $M\to\infty$.
\end{thm}
\begin{proof}
Omitted. The result is obtained following the same derivation as \cite[Theorem 2.1]{RupDavWan:J94}
in a straight-forward way.
\end{proof}
As a result from Theorems \ref{thm:asymp 0th} and \ref{thm:asymp 1st },
if one chooses the window size $b$ according to $M$ as $b=M^{-p}$,
then the bias asymptotically converges to $0$ at a rate $C_{2}M^{-2p}$,
where the coefficient $C_{2}$ depends on the choice of interpolation
methods (e.g., (\ref{eq:sol-0th}) and (\ref{eq:firstsol})), and
the variance converges to $0$ at a rate $\frac{C_{1}\sigma^{2}}{f(\bm{c}_{ij})}M^{-(1-2p)}$,
which is identical for both interpolation methods. In particular,
under $p=1/6$, the \Ac{mse}, which equals to the squared-bias plus
the variance, converges to $0$ at a fast rate, where both the squared-bias
and the variance terms scale as $O(M^{-2/3})$.

In addition, we discover some useful analytical insights from Theorems
\ref{thm:asymp 0th} and \ref{thm:asymp 1st } as follows.

\emph{Advantage provided by the first order method:} The advantage
of the first order method over the zeroth order counterpart is observed
from the bias of the estimator. Specifically, the first term in the
bias coefficient $\vartheta_{1}(\bm{c}_{ij})$ under the zeroth order
method disappears in the bias under the first order interpolation
method. As a result, the first order method is less affected by the
sensor distribution $f(\bm{z})$. For example, under non-uniform sensor
distribution, the term in $\vartheta_{1}(\bm{c}_{ij})$ yields a non-zero
bias, but such a bias does not appear in Theorem \ref{thm:asymp 1st }
under the first order method.

\emph{Impact from the sensor distribution and the propagation map:}
The bias of the estimator depends on the sensor distribution $f(\bm{c}_{ij})$
under the zeroth order method but not for the first order method.
The bias of both zeroth and first order method tend to be larger in
the area where the propagation function is with larger curvature $\partial^{2}\rho(\bm{c}_{ij})/\partial u_{x}^{2}$,
$\partial^{2}\rho(\bm{c}_{ij})/\partial u_{y}^{2}$. In addition,
the bias under the zeroth order method also depends on the gradient
of the sensor density $f(\bm{z})$. The analytical result matches
with the intuition that more sensors are generally required in the
region where the propagation map has larger variation.

\emph{Intuition on the variance:} The variance of the interpolation
error is not sensitive to model complexity as it remains the same
for both zeroth order method and first order method. The variance
is inversely proportional to $Mb^{2}f(\bm{c}_{ij})$ which represents
the number of sensors within the window centered at point $\bm{c}_{ij}$
with radius $b$. This implies that for the area with dense measurements,
the interpolation error of that area tends to have small variance.

\emph{The bias and variance trade off:} It is observed from the results
of Theorem \ref{thm:asymp 0th} and Theorem \ref{thm:asymp 1st }
that the bias is proportional to the square of the window size $b$
but the variance is inversely proportional to the square of $b$.
This implies that there might exist an optimal window size $b$, thus
the MSE of the estimator $\hat{\rho}$ attains smallest.

\subsection{Window Size Optimization \label{subsec:Bandwidth-Optimization-with}}

Here, we explore two methods to optimize $b$.

\subsubsection{Leave One Out cross validation}

A natural way to choose $b$ is to employ the Leave One Out cross
validation \cite{VehAkiAnd:J17}. The method minimizes the MSE of
the interpolation by solving
\begin{equation}
\underset{b}{\text{minimize}}\quad\frac{1}{M}\sum_{m=1}^{M}\left(\gamma_{m}-\hat{\rho}_{-m}(\bm{z}_{m})\right){}^{2}\label{eq:optimize b}
\end{equation}
where $\hat{\rho}_{-m}(\bm{z}_{m})$ is obtained in the same way as
$\hat{\rho}(\bm{z}_{m})$ in (\ref{eq:LS}) based on the $M-1$ samples
from $\left\{ (\bm{z}_{j},\gamma_{j})\right\} _{j=1}^{M}$ that exclude
the $m$th one.

However, Leave One Out cross validation is computationally expensive
when the data set is big with large $M$ as it requires fitting the
model $\hat{\rho}_{-m}$ $M$ times. Moreover, the estimated minimizer
$\hat{b}$ obtained from (\ref{eq:optimize b}) is mismatched with
the true minimizer on window size as shown in Fig.~\ref{fig:bandopt}.
The possible reason is that, the $m$th sensor plays the most significant
role to predict $\gamma_{m}$ at location $\bm{z}_{m}$, but the $m$th
sensor is excluded from estimating $\hat{\rho}_{-m}(\bm{z}_{m})$
in the Leave One Out cross validation.

\subsubsection{Analytical minimum \Ac{mse} method\label{subsec:Analytical-minimum-}}

We propose to optimize the window size $b$ by minimizing the \Ac{mse}
of the interpolation using the analytical results from Theorems \ref{thm1:0-th bv}
and \ref{thm2:1st bv}. The problem is formulated as follows:

\begin{equation}
\underset{b}{\mbox{minimize}}\quad\frac{1}{|\Omega|}\sum_{(i,j)\in\Omega}\mu_{ij}^{2}+\nu_{ij}^{2}\label{eq:optb}
\end{equation}
where $\Omega$ is defined in (\ref{eq:model-sparse-observation})
which depends on $b$, and $\mu_{ij}$ and $\nu_{ij}$ are respectively
the bias and variance derived from Theorems \ref{thm1:0-th bv} and
\ref{thm2:1st bv}. Specifically, $\mu_{ij}$ and $\nu_{ij}$ are
given as follows:
\begin{itemize}
\item Under the zeroth order interpolation (\ref{eq:sol-0th})
\begin{equation}
\mu_{ij}=\sum_{m=1}^{M}\hat{\bm{\beta}}^{\text{T}}(\bm{c}_{ij})(\bm{z}_{m}-\bm{c}_{ij})\bar{w}_{m}\label{eq:mu zeroth}
\end{equation}
\begin{equation}
\nu_{ij}^{2}=\sum_{m=1}^{M}\bar{w}_{m}^{2}(\bm{c}_{ij})\sigma^{2}\label{eq:nu zeroth}
\end{equation}
where $\hat{\bm{\beta}}(\bm{c}_{ij})$ is the gradient estimated from
(\ref{eq:LS-1-order}).
\item Under the first order interpolation (\ref{eq:firstsol})
\begin{equation}
\mu_{ij}=\frac{1}{2}\left[\tilde{\bm{W}}_{ij}^{-1}\tilde{\bm{D}}_{ij}\bm{W}_{ij}\text{diag}\{\bm{D}_{ij}^{\text{T}}\hat{\bm{\Psi}}{}_{ij}\bm{D}_{ij}\}\right]_{(1,1)}\label{eq:mu first}
\end{equation}
\begin{equation}
\nu_{ij}^{2}=\sigma^{2}[\tilde{\bm{W}}_{ij}^{-1}\left(\tilde{\bm{D}}_{ij}\bm{W}_{ij}\bm{W}_{ij}^{\text{T}}\tilde{\bm{D}}_{ij}^{\text{T}}\right)\tilde{\bm{W}}_{ij}^{-1}]_{(1,1)}\label{eq:nu first}
\end{equation}
where $\hat{\bm{\Psi}}$ is the estimated Hessian obtained from solving
problem (\ref{eq:LS}) under a second order model.
\end{itemize}
\begin{figure}
\subfigure{\includegraphics[width=0.5\columnwidth]{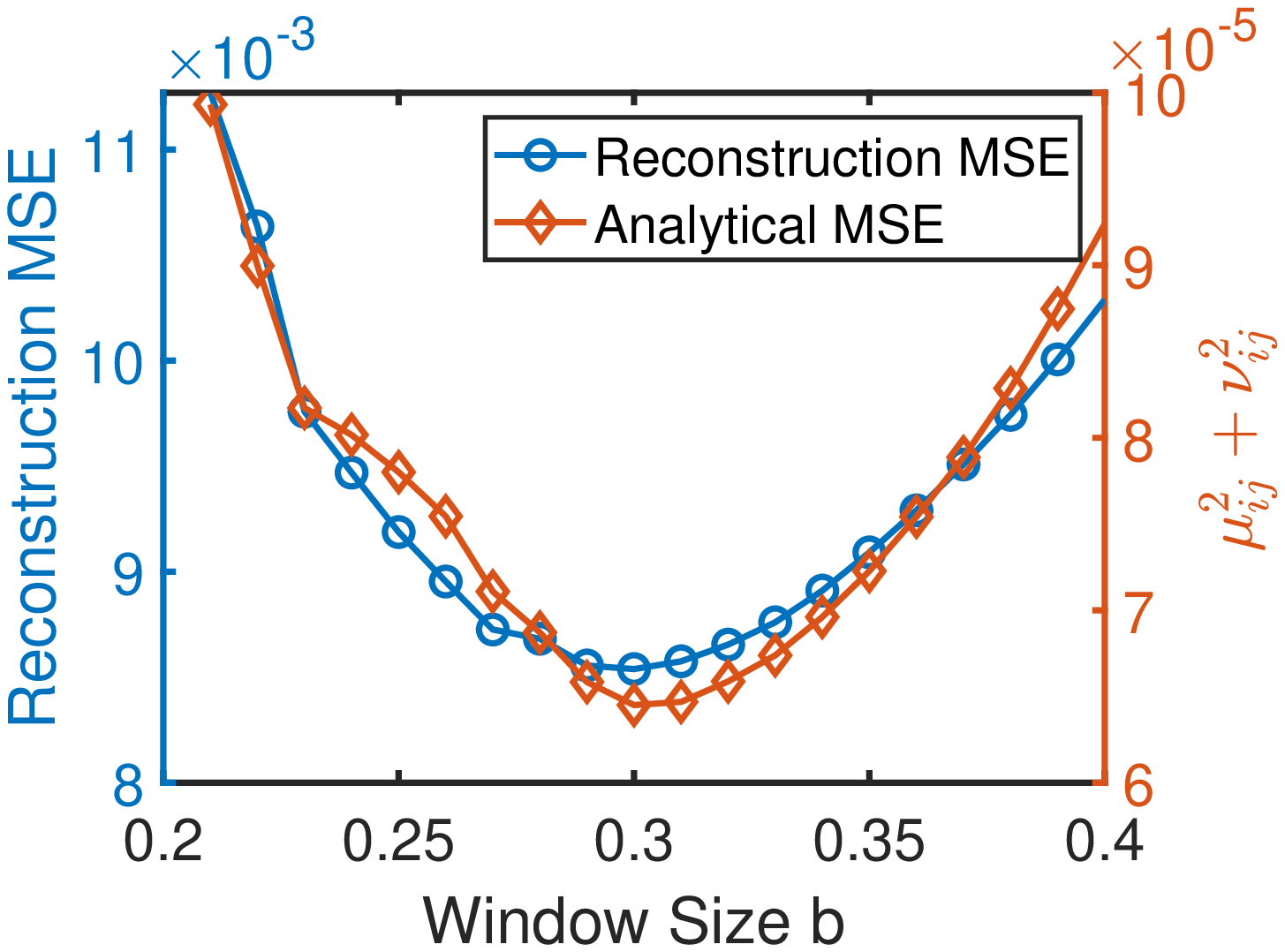}}\subfigure{\includegraphics[width=0.5\columnwidth]{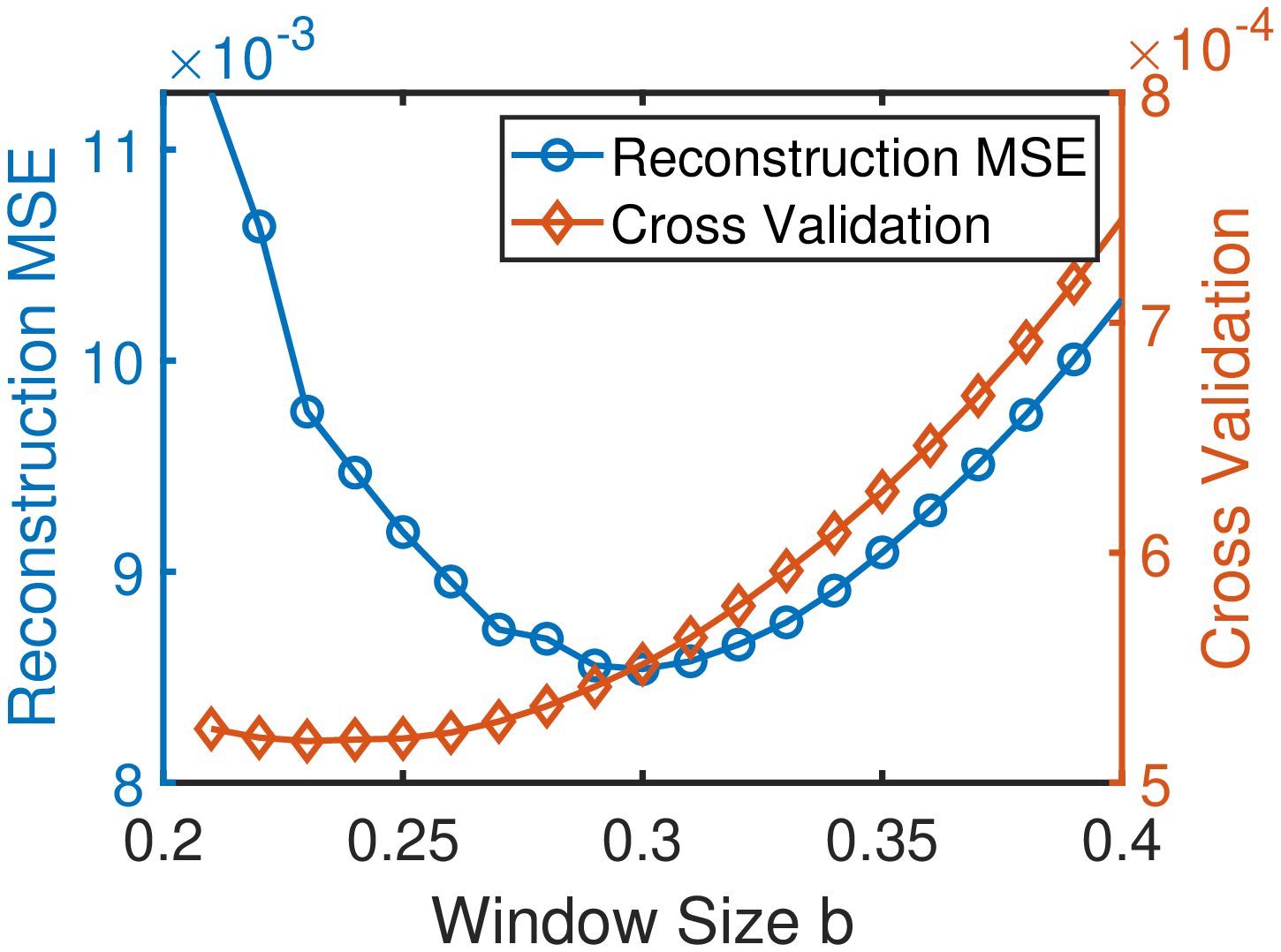}}

\qquad  \ \ \,\, \,\,\,\,\,\, \,\quad \footnotesize (a) \qquad \qquad \qquad \qquad \qquad \qquad \qquad  \, \footnotesize(b)

\caption{\label{fig:bandopt}(a) The MSE of the propagation map reconstruction
and the analytical MSE of the interpolation based on (\ref{eq:optb});
the minimizers $\hat{b}$ are matched; (b) The MSE of the propagation
map reconstruction and the cost from cross validation based on (\ref{eq:optimize b});
the minimizers $\hat{b}$ are mismatched.}
\end{figure}

In addition, the parameter $M_{0}$ in (\ref{eq:model-sparse-observation})
can be specified as $M_{0}=4$ for the zeroth order interpolation
method, since there are $3$ parameters to estimate including $\hat{\alpha}$
and the gradient $\hat{\bm{\beta}}(\bm{c}_{ij})$ for evaluating (\ref{eq:mu zeroth})
and (\ref{eq:nu zeroth}), and we need one additional measurement
since the support of the kernel in (\ref{eq:support}) has a radius
of $b$ and the measurement on the edge of the window will receive
a zero weight. Similarly, we may configure $M_{0}=7$ for the first
order interpolation, since there are $6$ parameters to estimate for
(\ref{eq:mu first}) and (\ref{eq:nu first}).

Fig. \ref{fig:bandopt} shows a numerical example on the MSE of propagation
map reconstruction under various choices of window size $b$. In the
experiment, we use the same propagation model as in Section \ref{sec:Numerical-Results}.
The dimension is $N=30$ and the number of sensors is $M=400.$ The
MSE of the propagation map reconstruction is calculated through $1/N^{2}||\bar{\bm{H}}-\bm{H}||_{F}^{2}$,
where $||\text{\ensuremath{\cdot}}||_{F}$ represents Frobenius norm,
and $\bar{\bm{H}}$ is the reconstructed propagation map based on
solving the matrix completion problem (\ref{eq:NNM-t}) in which the
$\hat{H}_{ij}$ is constructed from the local interpolation and $\mathcal{I}_{ij}$
is the uncertainty interval based on the local interpolation error.
More details to be discussed in Section \ref{sec:global reconstruction}.
The $\mu_{ij}$ and $\nu_{ij}$ are calculated according to (\ref{eq:mu first})
and (\ref{eq:nu first}). Fig.~\ref{fig:bandopt} (a) shows the reconstruction
MSE of matrix completion and the analytical MSE derived in (\ref{eq:optb})
in terms of $b$. It is observed that the minimizers $\hat{b}$ of
both curves are close. Fig.~\ref{fig:bandopt} (b) shows the reconstruction
MSE of matrix completion and the MSE in (\ref{eq:optimize b}) from
cross validation. It is observed that the minimizers are significantly
mismatched, confirming that the proposed analytical minimum MSE method
works better in optimizing the window size $b$. A slight adjustment
of window size $b$ in a $\pm0.1$ range may lead to $20$\% performance
difference. It is thus critical to optimize $b$.

\section{Global Reconstruction via Matrix Completion\label{sec:global reconstruction}}

In this section, we study matrix completion methods for propagation
map reconstruction that exploits the global structure of the propagation
map. We first numerically verify the low rank property of a propagation
map, and then, we develop uncertainty-aware matrix completion methods
that exploit the uncertainty from the local reconstruction.

\subsection{The Low Rank Property\label{subsec:lowrank simulation}}

\begin{figure}
\includegraphics{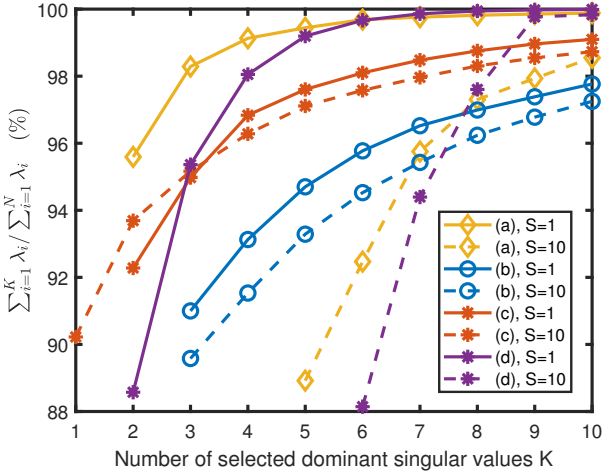}\caption{\label{fig:singular value}Percentage of the sum of the first $K$
dominant singular values over the sum of all singular values for a
$100$ by $100$ propagation matrix.}
\end{figure}
It is believed that the matrix representation of the propagation map
as constructed in Section \ref{subsec:Field-Reconstruction-via} is
likely low rank. The intuition is as follows: First, a propagation
map tends to have a unimodal structure where the longer the distance
away from the source, the smaller the RSS; second, there is usually
significant spatial correlation and hence $\rho(\bm{z})$ is smooth,
leading to the low rankness.

One special example to understand the low rank property is the exponential
2-D propagation map given by $g(d(\bm{s},\bm{z}))=e^{-d^{2}/\sigma}$
over $\bm{z}=(x,y)$, which can be decomposed as the product of two
1-D functions $g(d(\bm{s},\bm{z}))=u(x)v(y)$. This is due to the
fact that $d(\bm{s},\bm{z})^{2}=(x-s_{1,1})^{2}+(y-s_{1,2})^{2}$
and $e^{-(x^{2}+y^{2})/\sigma}=e^{-x^{2}/\sigma}e^{-y^{2}/\sigma}=u(x)v(y)$.
As a result, the matrix representation of the propagation map can
be written as the outer product of two vectors $\bm{H}=\bm{u}\bm{v}^{\text{T}}$,
indicating that $\bm{H}\in\mathbb{R}^{N\times N}$ is always rank-1
regardless of the matrix dimension $N$.

To numerically study the rank property of more propagation maps, we
select four common forms of propagation models and pick the parameters
that \emph{maximize} the rank of the propagation matrix $\bm{H}$
using a brute force search. The resulting models are listed as follows:
\begin{itemize}
\item[(a)]  $g(d)=\alpha d^{-\beta},\ \beta=2.2,\ h=0.09$
\item[(b)]  $g(d)=\alpha\mbox{exp}(-d^{\beta}),\ \beta=1.8,\ h=0.05$
\item[(c)]  $g(d)=\alpha-\beta\times\mbox{log}_{10}d,\ \beta=1.8,\ h=0.01$
\item[(d)]  $g(d)=\alpha d^{-\gamma}\beta^{-d},\ \beta=2.8,\ \gamma=1.5,\ h=0.01$
\end{itemize}
where $d=\sqrt{x^{2}+y^{2}+h^{2}}$ represents the distance from the
source at the origin, $(x,y)$ is the coordinate of the grid cell,
$h$ is the elevation, and the parameters $\alpha$ are selected to
normalize the total energy in the area. Note that models (a), (b)
and (c) correspond to energy in linear scale, exponential scale, and
log-scale, respectively, in radio propagation; in addition, model
(d) corresponds to a model for underwater acoustic propagation \cite{BRELeoMak:B04}.
The shadowing component is with the same setting as in Section \ref{sec:Numerical-Results}.

Fig.~\ref{fig:singular value} shows the percentage of the sum of
the first $K$ dominant singular values over the sum of all singular
values $\lambda_{i}$ for $N=100$. It clearly demonstrates the low
rank property, where $5$ dominant eigen-mode combined already reconstruct
over $95$\% energy of the propagation map for the single source scenario,
$S=1.$ For the scenario where $S=10$ sources are randomly placed
in the area and each source emits signal with the power randomly generated
from an exponential distribution with the rate parameter $\kappa=1$,
it is observed that there is still a low rank property for the aggregated
propagation map, and the rank increases much slower than the increasing
of the number of sources.

\subsection{Uncertainty-aware Matrix Completion}

The remaining challenge is to design new matrix completion methods
that exploit the knowledge that the observed entries may have different
uncertainty according to the local interpolation. We explore two uncertainty-aware
matrix completion formulations to demonstrate the core idea.

\subsubsection{Nuclear norm minimization with trust region (NNM-t)}

We introduce \emph{trust regions} in the following nuclear norm minimization
problem to exploit the local uncertainty in matrix completion
\begin{equation}
\begin{aligned}\underset{\bm{X}\in\mathbb{R}^{N\times N}}{\text{minimize}} & \quad\|\bm{X}\|_{*}\\
\text{subject to} & \quad X_{ij}-\hat{H}_{ij}\in\mathcal{I}_{ij},\qquad\forall(i,j)\in\Omega
\end{aligned}
\label{eq:NNM-t}
\end{equation}
where $\mathcal{I}_{ij}$ is the interval to specify the trust region
of the estimate $\hat{H}_{ij}$ obtained from the local reconstruction
in Section \ref{sec:Local-Polynomial-Regression}, and $\Omega$ is
defined in (\ref{eq:model-sparse-observation}).

The trust region $\mathcal{I}_{ij}$ for the estimation $\hat{H}_{ij}$
in (\ref{eq:NNM-t}) can be constructed using Theorems \ref{thm1:0-th bv}\textendash \ref{thm:asymp 1st }.
Specifically, Theorems \ref{thm:asymp 0th} and \ref{thm:asymp 1st }
suggest that the interpolation error $\xi_{ij}=\hat{H}_{ij}-\rho(\bm{c}_{ij})$
asymptotically follows a Gaussian distribution, where the mean $\mu_{ij}$
and variance $\nu_{ij}^{2}$ can be practically computed via (\ref{eq:mu zeroth})\textendash (\ref{eq:nu first}).
Then, using Gaussian approximation, the trust region $\mathcal{I}_{ij}$
with a $(1-\delta)$-confident level, \emph{i.e.}, $\mathbb{P}\{\hat{H}_{ij}-\rho(\bm{c}_{ij})\in\mathcal{I}_{ij}\}=1-\delta$,
can be approximately constructed as 
\begin{equation}
\mathcal{I}_{ij}=\Big(\text{\ensuremath{\mu_{ij}}}\pm\Phi^{-1}\Big(1-\frac{\delta}{2}\Big)\nu_{ij}\Big)\label{eq:confident-interval-I_ij}
\end{equation}
where $\Phi(x)$ is the cumulative distribution function of the standard
Gaussian distribution, and the notation $(a\pm b)$ represents an
interval $(a-b,a+b)$. The overall steps can be referred to Algorithm
\ref{alg:Propagation-field-reconstruction}.
\begin{algorithm}
Initialize $M$, $N$, $\bm{z}_{m}$, $\gamma_{m}$, $\delta$, $b_{\min}$,
$b_{\max}$ and interpolation order $r$.
\begin{enumerate}
\item Distance-weighted least-squares regression: Solving (\ref{eq:LS})
with $(r+1)$th order model.
\item Bias and variance approximation: $\mu_{ij}$ and $\nu_{ij}$.
\begin{enumerate}
\item If $r=0$, approximate (\ref{eq:0th-bias}), (\ref{eq:0th-variance})
with (\ref{eq:mu zeroth}), (\ref{eq:nu zeroth}).
\item If $r=1$, approximate (\ref{eq:1st bias}), (\ref{eq:1st-variance})
with (\ref{eq:mu first}), (\ref{eq:nu first}).
\end{enumerate}
\item Optimize $b$: Find the minimizer $\hat{b}$ of (\ref{eq:optb}) subject
to $b_{\mbox{\scriptsize min}}\leq b\leq b_{\mbox{\scriptsize max}}$.
\item Estimate $\hat{H}_{ij}$. Solving (\ref{eq:LS}) to get $\hat{\alpha}(\bm{c}_{ij})$
and $\hat{H}_{ij}=\hat{\alpha}(\bm{c}_{ij})$.
\begin{enumerate}
\item If $r=0$, 
\[
\hat{H}_{ij}=\bigg(\sum_{m=1}^{M}w_{m}(\bm{c}_{ij})\gamma_{m}\bigg)\Big/\sum_{m=1}^{M}w_{m}(\bm{c}_{ij}).
\]
\item If $r=1$, 
\[
\hat{H}_{ij}=\left[(\tilde{\bm{D}}_{ij}\bm{W}_{ij}\tilde{\bm{D}}_{ij}^{\text{T}})^{-1}\tilde{\bm{D}}_{ij}\bm{W}_{ij}\bm{\gamma}\right]_{(1,1)}.
\]
\end{enumerate}
\item Repeat step $1$.
\item Confidence interval $\mathcal{I}_{ij}$ construction: $\mathcal{I}_{ij}=\big(\mu_{ij}\pm\Phi^{-1}(1-\frac{\delta}{2})\nu_{ij}^{\frac{1}{2}}\big)$.
\item Nuclear norm matrix completion. Get the recovered matrix $\bar{\bm{H}}$
through solving (\ref{eq:NNM-t}).
\end{enumerate}
\caption{\label{alg:Propagation-field-reconstruction}Propagation map reconstruction
under NNM-t}
\end{algorithm}

\subsubsection{Weighted ALS (wALS)\label{subsec:Weighted-ALS}}

We design a wALS formulation to exploit the uncertainty of the observation
based on the conventional ALS algorithm developed in \cite{DavMar:J12}.
We propose to assign weights, which are inversely proportional to
the standard variance of the interpolation error, to the matrix observation.
The weights $\bm{w}$ are chosen with the intuition that the estimate
with smaller variance is believed to have higher accuracy and thus
plays a more significant role in the matrix completion problem.

Let $M'=|\Omega|$ be the number of observed entries via local reconstruction
from interpolation. For the $n$th element in $\Omega$ that corresponds
to the $(i,j)$th entry of the matrix, construct an observation $y_{n}=\hat{H}_{ij}-\mu_{ij}$,
and construct a weight $w_{n}=\nu_{ij}^{-1}$. Let $\mathcal{A}:\mathbb{R}^{N\times N}\to\mathbb{R}^{M'}$
be a sensing operator that maps the $(i,j)$th element of a matrix
$\bm{X}$ to the corresponding $n$th element of the vector $\bm{y}$.
The wALS problem is formulated as follows
\begin{equation}
\begin{aligned}\underset{\bm{X}\in\mathbb{R}^{N\times N}}{\text{minimize}} & \quad||\bm{w}\circ(\bm{y}-\mathcal{A}(\bm{X}))||_{2}\\
\text{subject to} & \quad\bm{X=\bm{LR}}\\
 & \quad\bm{L}\in\mathbb{R}^{N\times p}\\
 & \quad\bm{R}\in\mathbb{R}^{p\times N}
\end{aligned}
\label{eq:wALS}
\end{equation}
where `$\circ$' means element-wise product.

To numerically evaluate the performance of the proposed uncertainty-aware
matrix completion algorithms with their conventional counterparts,
consider the same experiment configuration in Section \ref{sec:Numerical-Results}
with number of sources $S=1$, number of measurements $M=40-200$
for reconstructing a propagation map with resolution $N=30$. For
a conventional NNM scheme, it solves a similar NNM problem except
that it replaces the constraint in (\ref{eq:NNM-t}) with $|X_{ij}-\hat{H}_{ij}|\leq\epsilon$,
where $\epsilon=2/M$ is the best parameter we found via cross validation.
For a conventional ALS, the weighting vector $\bm{w}$ in (\ref{eq:wALS})
is $\ensuremath{\bm{1}}$. For the proposed methods, the confident
level parameter $\delta=0.05$ is chosen for the NNM-t method, and
the rank parameter $p=1$ is chosen for both wALS and the conventional
ALS. Fig.~\ref{fig:ALS} compares the \ac{mse} in propagation map
reconstruction under different matrix completion algorithms. The results
demonstrate that both of the proposed NNM-t and wALS methods outperform
their conventional counterparts, which only exploit the local reconstruction
$\hat{H}_{ij}$ from Section \ref{sec:Local-Polynomial-Regression},
but not the bias and variance derived from our analysis. The results
confirm the advantage of explicitly exploiting the observation uncertainty
for enhancing the reconstruction performance.
\begin{figure}
\includegraphics{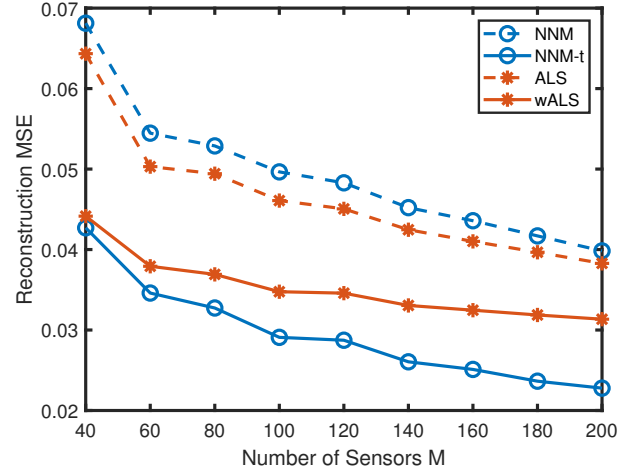}\caption{\label{fig:ALS} Reconstruction MSE for $S=1$ and $N=30$. The uncertainty-aware
scheme works.}
\end{figure}

\subsection{Identifiability of the Matrix Completion\label{subsec:Identifiability-of-the}}

Recall that under the sensor-aware approach, the observation set $\Omega$
is formed according to the deterministic sensor deployment $\{\bm{z}_{m}\}$
and the windows size parameter $b$. Thus, it is important to understand
whether the matrix is completable based on $\Omega$. Although establishing
a rigorous approach to determining the best observation pattern of
$\hat{\bm{H}}$ with \emph{both} the interpolation performance and
the identifiability conditions \cite{CanRec:J12,PimDan:J16,CheChiFan:J20}
taken into account still remains an open question that goes beyond
the scope of this paper, we have obtained the following insights with
a simple strategy to guarantee the identifiability for matrix completion
with high probability.

First, it is found that, for uniformly random sensor deployment, the
proposed scheme tends to naturally meet the identifiability conditions
required by the low rank matrix completion \cite{CanRec:J12,PimDan:J16,CheChiFan:J20,CheBhoSan:J15}.
For sensor-aware deployment, our observations are as follows. For
a moderate to large $M$, recall that a subset of entries $\hat{H}_{ij}$
with $M_{0}$ measurements within a window size $b$ are constructed
via interpolation, where $M_{0}$ was chosen based on the identifiability
of the local polynomial regression model, and $b$ was optimized such
that the interpolated error for $\hat{H}_{ij}$ is minimized. This
step naturally forms a large number of observations $\hat{H}_{ij}$,
and such a phenomenon is consistent with the identifiability condition
of low rank matrix completion as observed from our experiments. For
a small $M$, the window size optimization tends to select a large
window size $b$ to meet the identifiability condition of the local
polynomial regression model (see Section \ref{subsec:Bandwidth-Optimization-with}),
and meanwhile, a large $b$ is also needed for reducing the local
interpolation error. As a result, a large number of observed entries
are constructed from the interpolation under a large window size $b$;
this, again, is consistent with the identifiability condition as observed
from our experiments.

Then, following the above insight, we introduce a simple strategy
in our implementation to guarantee the identifiability for matrix
completion with high probability. Specifically, we introduce a constraint
$b_{\text{\text{min}}}\leq b\leq b_{\text{max}}$ in (\ref{eq:optb})
to control the window size parameter $b$, where $b_{\text{\text{min}}}$
is chosen such that the number of elements in the observation set
$\Omega$ formed in (\ref{eq:model-sparse-observation}) is no fewer
that $\omega N^{2}$, where $\omega$ is a target sampling ratio which
is typically chosen to meet the asymptotic scaling rule $\omega N^{2}\geq CN\text{log}^{2}(N)$
\cite{CanPla:J10} with a sufficiently large parameter $C$ to guarantee
sufficient observations for small to medium $N$. In addition, after
obtaining $b$ from solving (\ref{eq:optb}), we adopt a simple identifiability
check for the observation set $\Omega$, where we require each row
and each column of the sparse observation matrix $\hat{\bm{H}}$ has
at least one entry in $\Omega$. If such a criterion is not satisfied,
then $b$ is repeatedly increased by $20$\% until such a criterion
is met. We have observed that, following this strategy, the proposed
matrix completion is robust for small to large $M$.

\subsection{Complexity}

The computation of the interpolation assisted matrix completion method
consists of two parts. First, the interpolation step has a computational
complexity $\mathcal{O}(M^{2}\omega N^{2})$ as $M$ is the number
of measurements and $N$ represents the resolution of the reconstructed
map. Here, $\omega<1$ is the interpolation rate, since only a subset
of grid points are constructed by the interpolation, and moreover,
the interpolation of each grid point requires only a \emph{subset}
of $M$ measurements nearby. More specifically, if a uniform sampling
approach is adopted to form $\Omega$ as discussed in Section II-B,
then only $CN\text{log}^{2}(N)$ grid cells are needed to be constructed.
As such, $\omega$ can be upper bounded by $CN\text{log}^{2}(N)/N^{2}$,
and the complexity of the interpolation step is $\mathcal{O}(M^{2}N\text{log}^{2}(N))$.

Second, for the NNM-t approach, the matrix completion step has complexity
$\mathcal{O}(CN\log^{3}(2N))$, where $C$ is a constant that captures
the \emph{incoherence} of the matrix and the accuracy requirement
of the matrix completion algorithm \cite{JaiNet:B15}. Thus, the total
complexity of the NNM-t method is $\mathcal{O}(M^{2}\omega N^{2}+CN\log^{3}(2N))$.
In particular, under uniform sampling, the total complexity can be
simplified to $\mathcal{O}(M^{2}N\text{log}^{2}(N))$ for large $N$
and $M$.

To benchmark, Kriging has a computational complexity of $\mathcal{O}(M^{3}N^{2})$
for constructing $N^{2}$ points \cite{SriDurMur:B10}, because, first,
Kriging needs to learn a semi-variogram, and second, Kriging requires
to use all the $M$ measurements to reconstruct every points. A numerical
comparison on the computational time is reported in Table \ref{tab:Computation-time-(seconds)}.

\section{Numerical Results\label{sec:Numerical-Results}}

\begin{figure}
\includegraphics{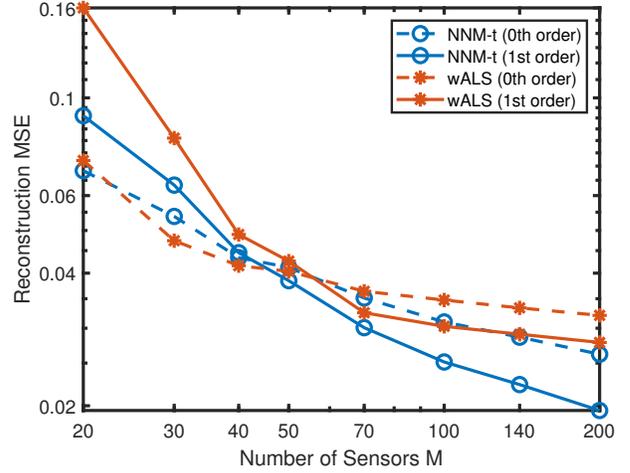}

\caption{\label{fig:cmp zeroth first} Reconstruction MSE for $S=1$ and $N=30$.
Zeroth order interpolation performs better for a small number sensors,
whereas, first order method prefers more sensors.}
\end{figure}

We adopt model (\ref{eq:model-propagation-field}) to simulate the
propagation map in an $L\times L$ area for $L=2$ kilometers, where
$g_{k}(d)=P_{k}d^{-1.5}A(f)^{-d}$, with parameter $A(f)=0.8$, corresponding
to an empirical energy field of underwater acoustic signal at frequency
$f=5$ kHz \cite{BRELeoMak:B04,StoMil:J09}, $d=\sqrt{x^{2}+y^{2}+h^{2}}$
represents the distance from the source, $(x,y)$ is the coordinate,
$h=400$ meters is the depth of interest for a multiple source scenario
we test later, and $h=1$ kilometers is for a single source scenario.
The parameter $P_{k}$ follows an independent exponential distribution
with rate parameter $\kappa=1$ to model the power emitted from source
$k$. The shadowing component in log-scale $\mbox{log}_{10}\zeta$
is modeled using a Gaussian process with zero mean and auto-correlation
function $\mathbb{E}\{\mbox{log}_{10}\zeta(\bm{z}_{i})\mbox{log}_{10}\zeta(\bm{z}_{j})\}=\sigma_{\text{s}}^{2}\mbox{exp}(-||\bm{z}_{i}-\bm{z}_{j}||_{2}/d_{\text{c}})$,
in which $d_{\text{c}}=200$ meters, $\sigma_{s}^{2}=1$. The sensors
are randomly and independently distributed in the area following a
uniform distribution. The measurement model in (\ref{eq:model-measurement})
is adopted, where the measurement noise $\epsilon$ is modeled as
a zero mean Gaussian random variable with standard deviation $\sigma=0.02$.

The MSE of the reconstructed propagation map is employed for performance
evaluation, which is calculated through $1/N^{2}||\bar{\bm{H}}-\bm{H}||_{F}^{2}$
as stated in Section \ref{subsec:Analytical-minimum-} .

The proposed interpolation assisted matrix completion is implemented
following Algorithm \ref{alg:Propagation-field-reconstruction} under
$M_{0}=7$ and confident level parameter $\delta=0.05$ as explained
in Sections \ref{subsec:Bandwidth-Optimization-with} and \ref{subsec:Weighted-ALS},
respectively. We evaluate two versions of implementation of the proposed
scheme. Scheme 1) NNM-t: The observation set $\Omega$ is constructed
using uniformly random sampling with $C=1.6$ as described in Section
\ref{subsec:Field-Reconstruction-via}, and $b_{\max}=0.35L$ to guarantee
that the interpolation only focuses on a local area. To circumvent
the possible singularity issue of estimating the interpolation parameters
via computing (\ref{eq:mu first}) and (\ref{eq:nu first}) for a
too small $b$, where $b$ serves as the kernel parameter in (\ref{eq:LS}),
we lower bound the kernel parameter as $b'=\max\{b,b_{\min}^{(i,j)}\}$
for each specific entry $(i,j)\in\Omega$, where $b_{\min}^{(i,j)}$
is the minimum radius from the grid center $\bm{c}_{ij}$ to include
$M_{0}$ sensor measurements. Scheme 2) NNM-t (adaptive): The observation
set $\Omega$ is constructed using sensor-aware sampling with an optimized
$b$ via solving (\ref{eq:optb}), where the parameter of identifiability
check is set as $C=2.5$.\footnote{The parameter $C=1.6$ and $2.5$ correspond to $60\%$ and over $90\%$
observation ratio, respectively, at matrix dimension $N=30$. Hence,
the NNM-t (adaptive) scheme requires high complexity in the interpolation
step.} In addition, for each grid cell, we optimize an individual window
size $b_{ij}$ by solving (\ref{eq:optb}) while dropping the summation.
The observation $\hat{H}_{ij}$ and the confident interval $\mathcal{I}_{ij}$
are computed based on $b_{ij}$.

The performance is compared with the following baselines that are
recently developed or adopted in related literature. Baseline $1$:
Kriging method \cite{ChoVin:J18,SatKoyFuj:J17}, the propagation map
is reconstructed by Universal Kriging $\hat{\rho}(\bm{c})=\hat{m}(\bm{c})+\epsilon(\bm{c})$,
where $\hat{m}(\bm{c})$ is the interpolated deterministic part, and
$\epsilon(\bm{c})$ is the interpolated residual which can be estimated
with an Ordinary Kriging method. Baseline $2$: $k$-nearest neighbor
local polynomial interpolation (k-LP) \cite{VerFunRaj:J16}, a first
order local polynomial regression method is used to estimate the \Ac{rss}
using $k$ nearest measurements, where $k$ is chosen through cross
validation. Baseline 3: On-grid low rank matrix completion (On-grid
LRMC) \cite{CanRec:J12}, based on the sensor topology, we first form
a refined on-grid measurement set as follows: if there is a sensor
in a grid cell, we re-sample the grid cell by sampling at the grid
center, and hence, there will be no discretization error. Note that
since this approach is not practical in real life, this scheme is
for performance benchmarking only. Then, the matrix is completed by
solving a conventional NNM problem. Baseline 4: Thin plate spline
(TPS) \cite{JuaGonGeo:J11}, a thin plate spline method is used to
construct $\hat{\rho}(\bm{c})$ from the scattered measurements $\{\bm{z}_{m},\gamma_{m}\}$.

\begin{figure}
\subfigure{\includegraphics[width=0.5\columnwidth]{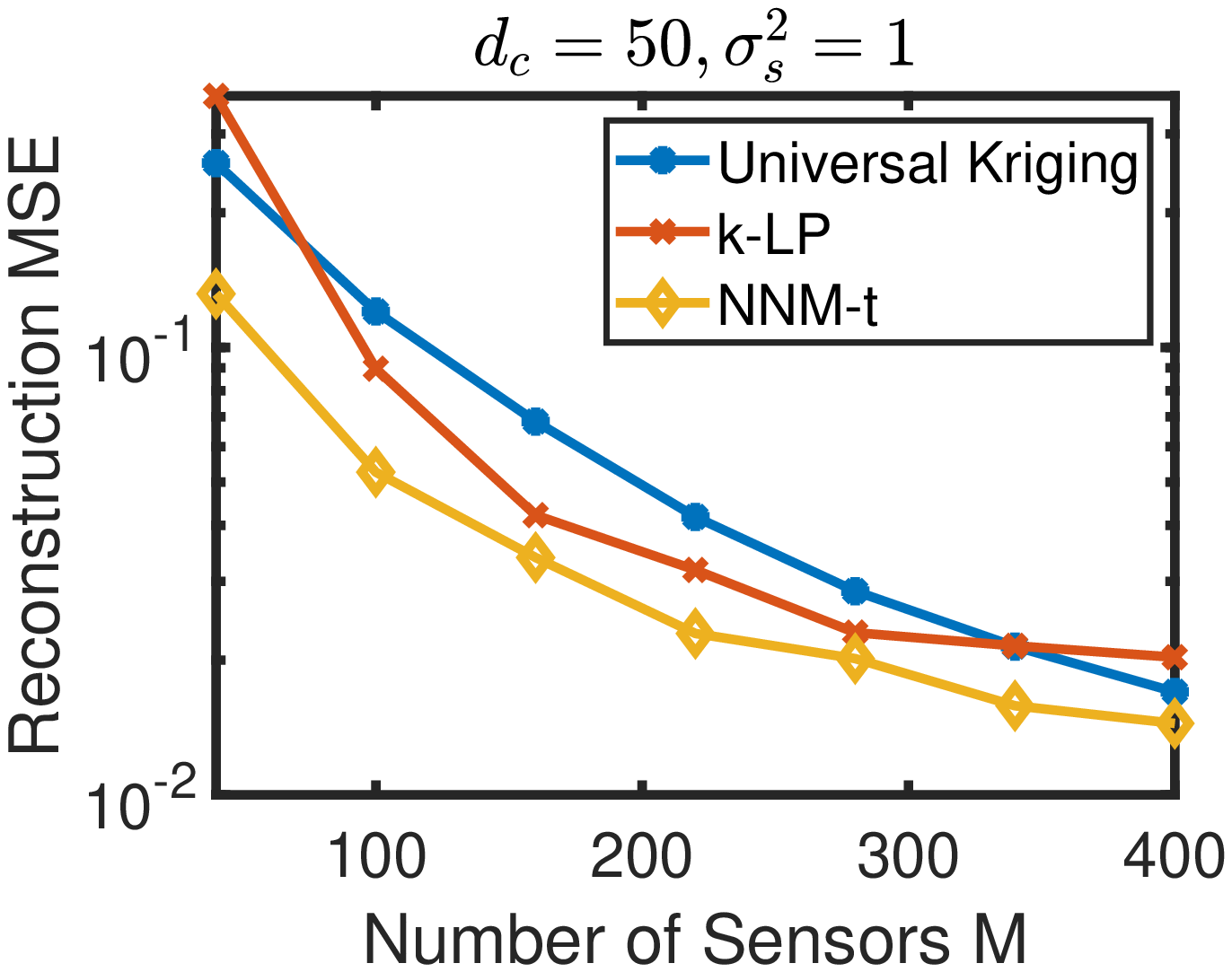}}\subfigure{\includegraphics[width=0.5\columnwidth]{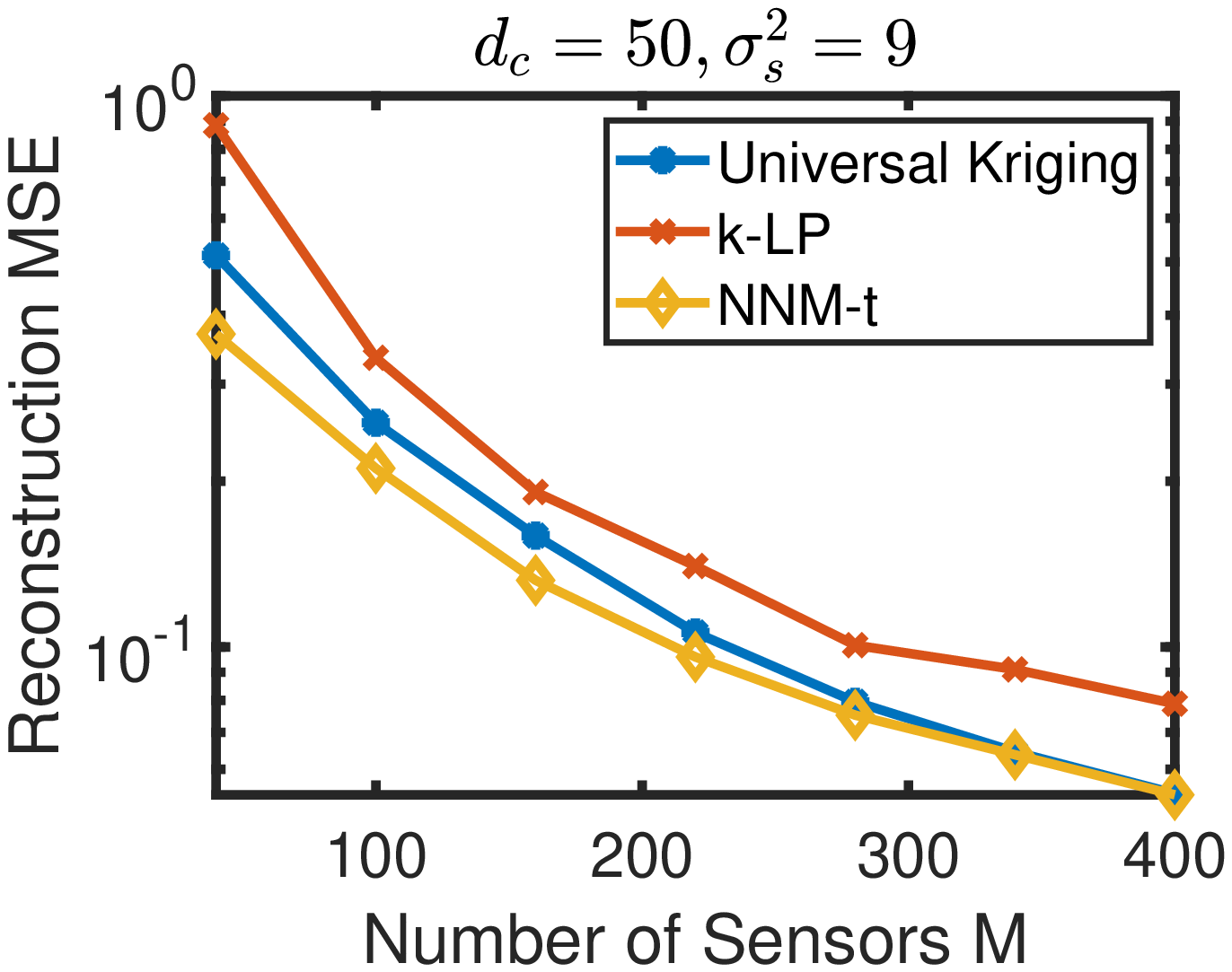}}\\

\subfigure{\includegraphics[width=0.5\columnwidth]{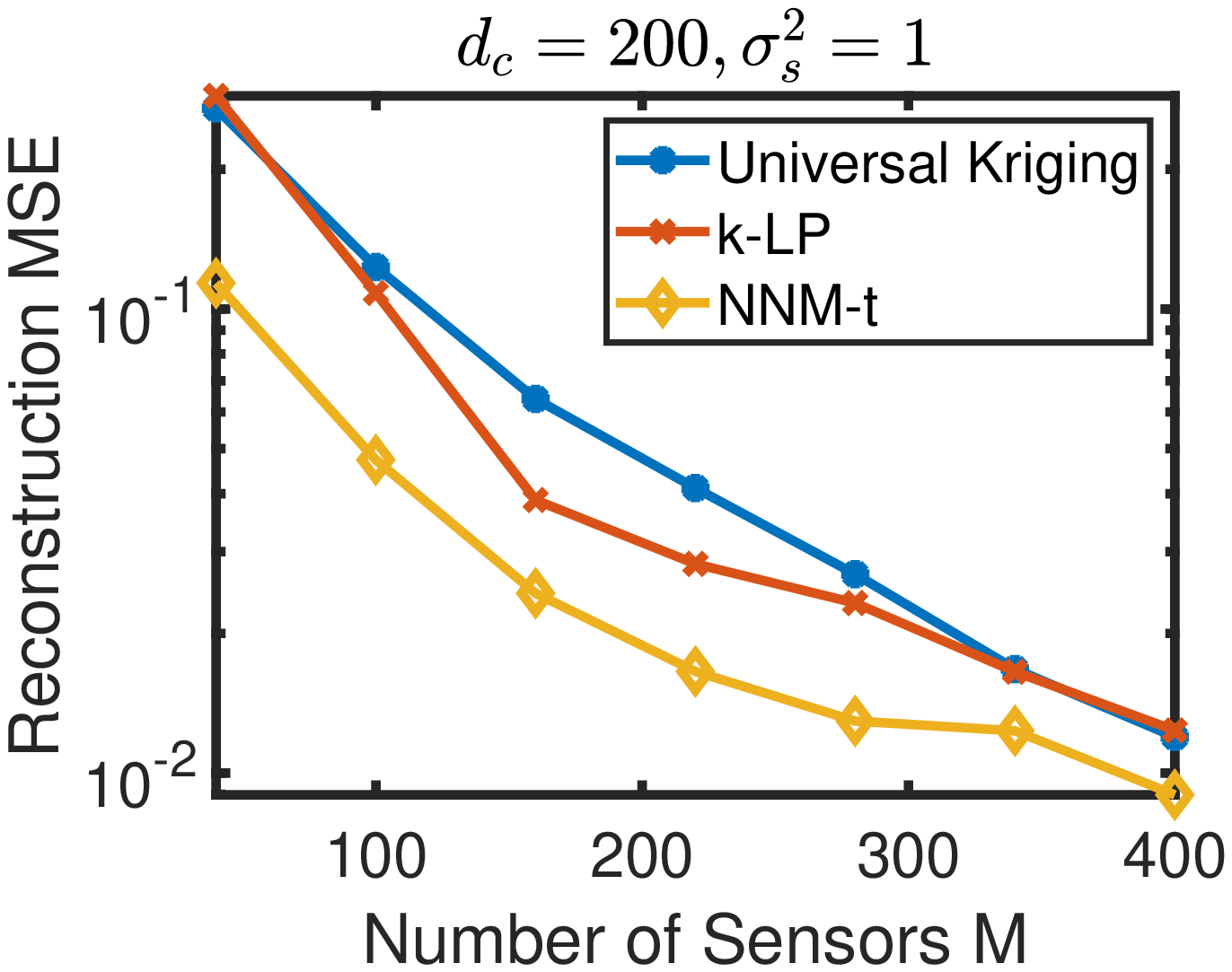}}\subfigure{\includegraphics[width=0.5\columnwidth]{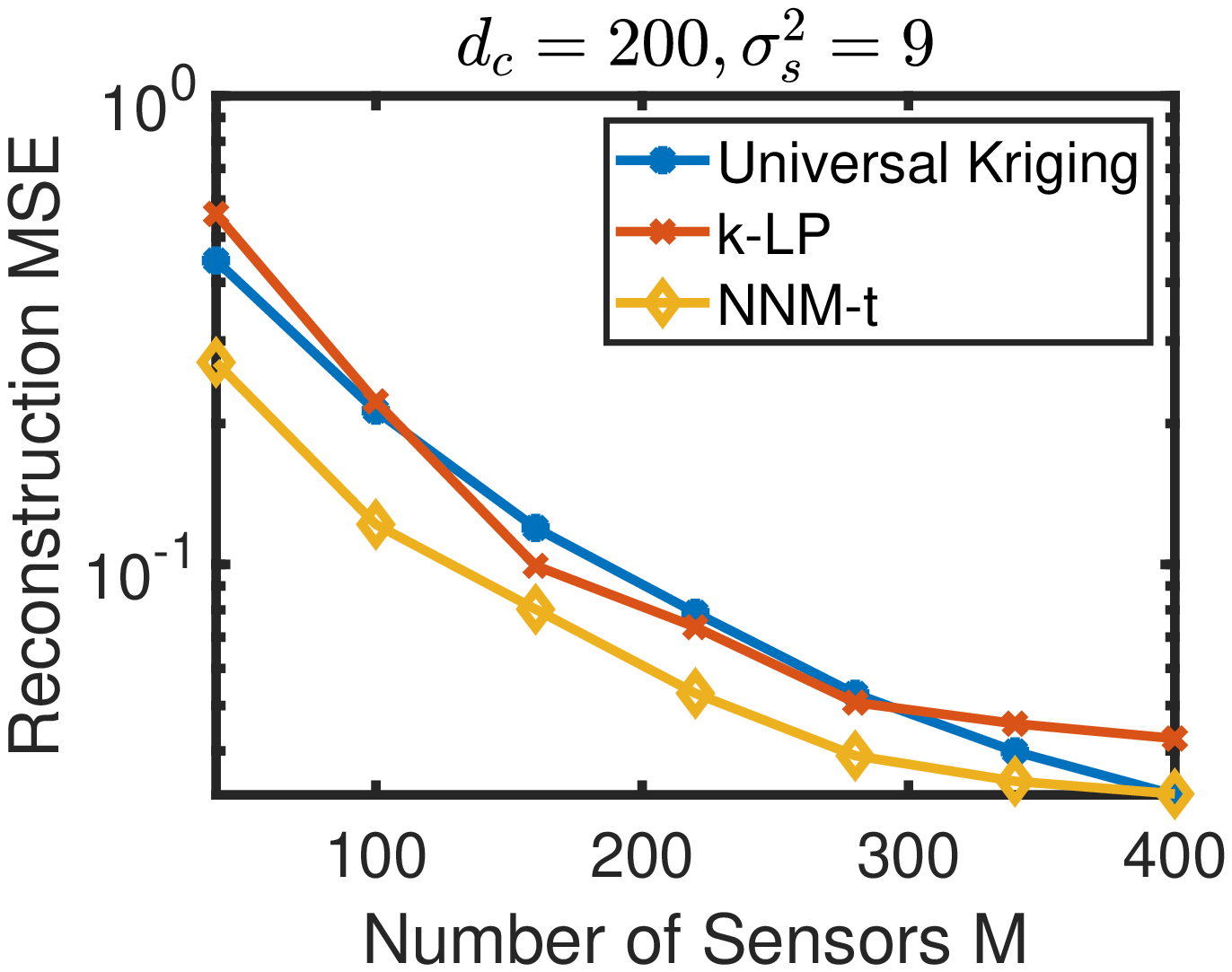}}

\caption{\label{fig:shadowing parameters}Performance under different shadowing
parameters.}
\end{figure}

\subsection{Zeroth Order or First Order Interpolation}

We evaluate the performance of the proposed schemes over different
interpolation models under different number of measurement samples.
For both NNM-t and wALS schemes, the zeroth order interpolation (\ref{eq:sol-0th})
and first order interpolation (\ref{eq:firstsol}) are used to construct
$\hat{H}_{ij}$. In a single source propagation map, $S=1$, and for
$N=30$, Fig.~\ref{fig:cmp zeroth first} shows that when there are
only a few measurements, the methods based on zeroth order interpolation
attain a smaller MSE than their counterparts based on first order
models. This is because zeroth order models have less parameters to
estimate. On the other hand, for medium to large number of measurements,
methods based on first order interpolation perform better, because
first order interpolation can substantially eliminate the bias in
the local interpolation as revealed in our asymptotic analysis stated
in Theorems \ref{thm:asymp 0th} and \ref{thm:asymp 1st }.

\subsection{Performance under Different Shadowing Parameters}

We quantify the propagation map reconstruction performance of the
proposed schemes under different shadowing parameters. The simulations
are performed under $d_{c}=\left\{ 50,200\right\} $ meters and $\sigma_{s}^{2}=\left\{ 1,9\right\} $
separately with the baselines are Universal Kriging and k-LP. The
experiment results shown in Fig.~\ref{fig:shadowing parameters}
reveal that the increasing of shadowing variance decreases the performance
of all baselines. The larger the correlation distance, the better
the performance.
\begin{figure}
\includegraphics{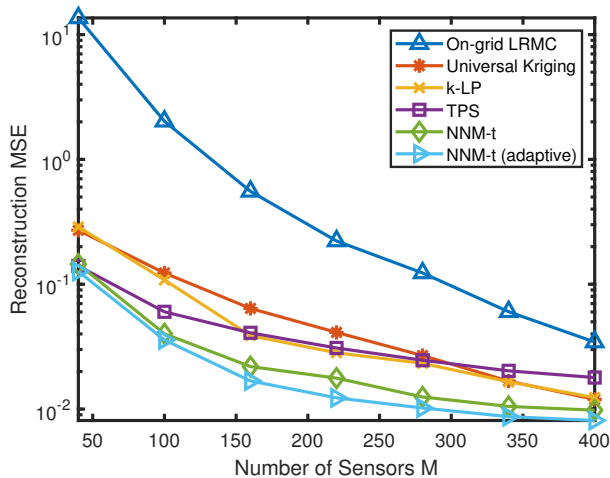}\caption{\label{fig:dimension}Reconstruction MSE for $S=3$ and $N=30$.}
\end{figure}
\begin{figure}
\includegraphics{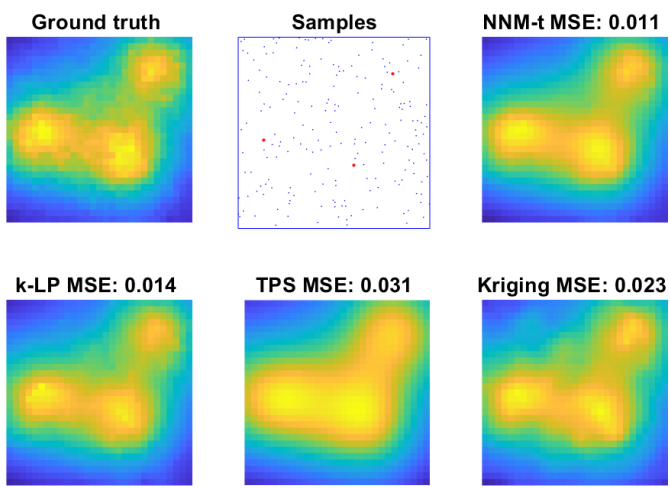}\caption{\label{fig:Ground-truth,-sampling}Ground truth, sampling pattern,
and the reconstruction MSE of different methods.}
\end{figure}

\subsection{Performance under Different Measurements}

We quantify the propagation map reconstruction performance of the
proposed schemes under different number of measurements $M=40-400$
with fixed resolution $N=30$ and number of sources $S=3$.

Fig.~\ref{fig:dimension} demonstrates the reconstruction \ac{mse}
versus the number of sensors $M$. A visual plot of a realization
of the propagation field, sampling pattern, and the reconstruction
results under $M=200$ are demonstrated in Fig.~\ref{fig:Ground-truth,-sampling}.
It is observed that the proposed NNM-t method can realize 10\%\textendash 50\%
MSE reduction from all baseline schemes under medium to large $M$,
which translates to a saving of roughly $1/3$ of sensor measurements
at a medium $M$. At small $M$, \emph{i.e.}, $M=40$, its performance
is similar to that of the TPS scheme. The NNM-t (adaptive) method
works better than the NNM-t method and outperforms all baseline schemes
because it optimizes the window size parameter $b_{ij}$ for each
grid cell in an adaptive way. The Universal Kriging and k-LP schemes
perform well at a large $M$, but their performance degrades at small
$M$, whereas, the TPS scheme works well at small $M$, but suffers
from substantial performance degradation at large $M$. Finally, the
On-grid LRMC scheme performs badly because the number of observations
$M=40$\textendash $400$ is substantially smaller than the total
number of entries $N^{2}=900$ to be recovered, and the matrix is
likely non-identifiable in this regime.

Table~\ref{tab:Computation-time-(seconds)} summarizes the computation
time based on our implementation using Matlab. The Universal Kriging
baseline is implemented by the \textquotedblleft mGstat\textquotedblright{}
function provided in the Geostatistical toolbox in Matlab. The NNM-t
problem is solved by CVX using an interior point method. The complexity
of k-LP and LRMC schemes is for benchmarking because they are the
building blocks of the proposed NNM-t. As observed in Table \ref{tab:Computation-time-(seconds)},
the proposed method has a lower complexity than Universal Kriging
at medium to large $M$ as expected from our analysis. 
\begin{table}
\caption{\label{tab:Computation-time-(seconds)}Computation time (seconds)
at different $M$ under $N=30$}

\begin{tabular}{c|ccccccc}
$M$ & 40 & 100 & 160 & 220 & 280 & 340 & 400\tabularnewline
\hline 
k-LP & 0.39 & 0.53 & 0.79 & 1.19 & 1.64 & 1.85 & 2.05\tabularnewline
LRMC & 6.63 & 6.39 & 6.60 & 6.76 & 6.45 & \multicolumn{1}{c}{7.11} & 6.91\tabularnewline
\hline 
U-Kriging & 2.51 & 5.77 & 10.63 & 19.49 & 28.32 & 40.37 & 55.76\tabularnewline
NNM-t & 8.26 & 6.10 & 6.67 & 8.57 & 6.72 & 6.98 & 8.73\tabularnewline
\end{tabular}
\end{table}

\begin{figure}
\includegraphics{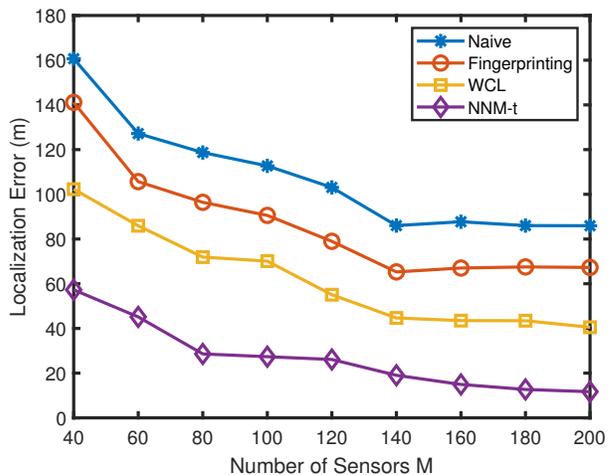}\caption{\label{fig:Source-localization-error}Source localization error for
$S=1$ and $N=30$.}
\end{figure}

The gain of the proposed method over a pure interpolation scheme can
be interpreted from the following two aspects. First, although there
is shadowing, the matrix of the propagation field is still relatively
low rank when choosing a large $N$ for the matrix dimension. In addition,
the proposed NNM-t method can dynamically adjust the windows size
$b$ such that there will always be enough observations to guarantee
the recoverability of the matrix completion. Second, a pure interpolation
method suffers from performance loss in the regions where measurements
are locally too sparse, typically, at the edge of the area of interest.
By contrast, the proposed interpolation assisted matrix completion
approach can fill in the missing values in these regions using the
global information constructed from the other part of the propagation
field.

\subsection{Application in RSS-based Source Localization}

Based on the reconstructed propagation map, we demonstrate the application
of propagation map reconstruction under NNM-t for $S=1$ to \ac{rss}-based
source localization. The localization algorithm based on matrix factorization
developed in \cite{CheMit:J17} is adopted. Specifically, based on
the reconstructed propagation map matrix $\bar{\bm{H}}$, the dominant
singular vectors $\bm{u}$ and $\bm{v}$ of $\bar{\bm{H}}$ are computed.
Then, the coordinates of the source are estimated via peak localization
of the dominant singular vectors $\bm{u}$ and $\bm{v}$ \cite{CheMit:J17}.

Three baseline schemes for localization are evaluated. Baseline $1$:
WCL \cite{WanUrrHanCab:J11}, we use $\hat{\bm{s}}_{\text{WCL}}=\sum_{m=1}^{M}w_{m}\bm{\bm{z}}_{m}/\sum_{m=1}^{M}w_{m}$
to estimate the location of the source, where $\text{\ensuremath{w_{m}}}=\gamma_{m}$
serves as the weight. Baseline $2$: naive method, the location of
the source is chosen as the location of the sensor with largest measurements.
Baseline 3: Fingerprinting \cite{JukMar:J15}, the propagation map
is first constructed using the Delaunay triangulation, and then, the
location with the largest RSS in the reconstructed propagation map
is selected as the location estimate of the source.

Fig.~\ref{fig:Source-localization-error} shows the RMSE of RSS-based
source localization based on the propagation map reconstructed from
the proposed NNM-t method. The RMSE is reduced by roughly a half and
both baseline schemes require $2$ times more sensors to achieve the
same localization accuracy.
\begin{figure}
\subfigure{\includegraphics[width=0.5\columnwidth]{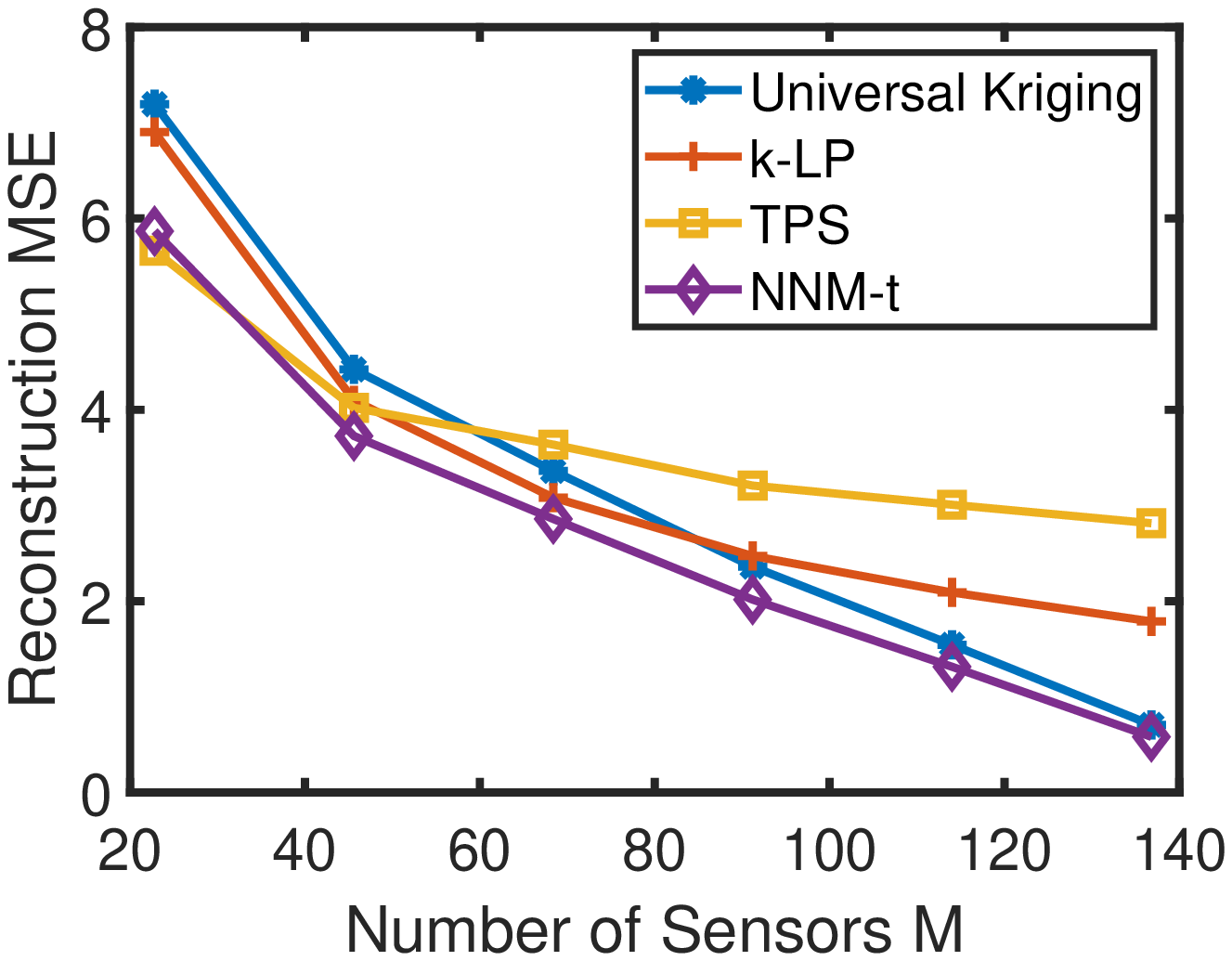}}\subfigure{\includegraphics[width=0.5\columnwidth]{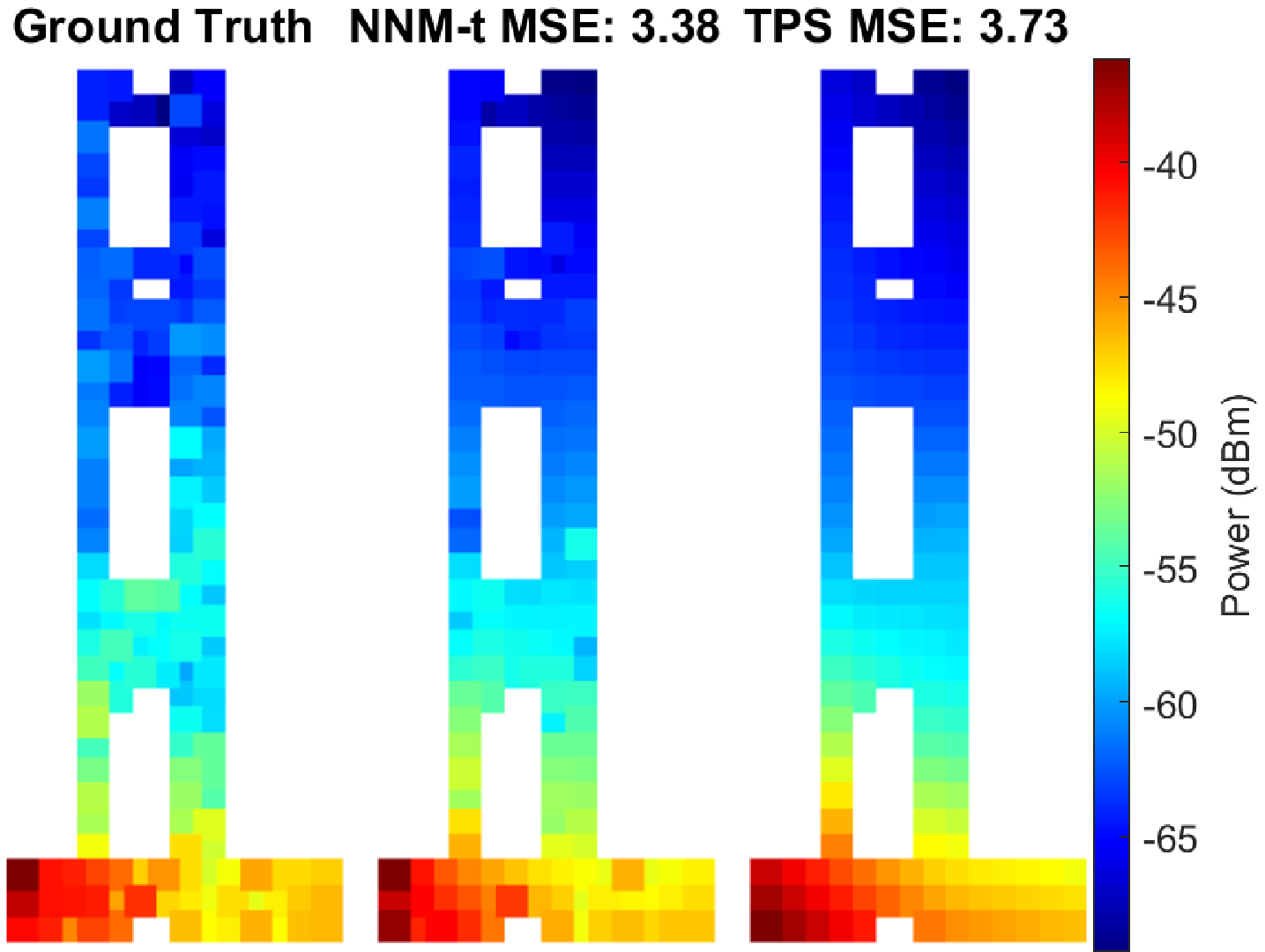}}

\qquad  \ \ \,\, \,\,\,\,\,\, \,\quad \footnotesize (a) \qquad \qquad \qquad \qquad \qquad   \, \footnotesize(b)

\caption{\label{fig:real data}The reconstruction performance evaluated on
real data set. (a) The reconstruction MSE versus the number of measurements
used for propagation map construction. (b) Visualization of the propagation
field in an indoor environment, showing the ground truth, the reconstruction
from NNM-t, and the reconstruction from TPS.}
\end{figure}

\subsection{Experiment on a Real Dataset}

We extend the experiment to a real dataset \cite{mannheim-compass-20080411}.
In this dataset, there are $M=166$ RSS measurements for each emitter.
The measurements are taken in a $14\times34\ \text{m}^{2}$ indoor
area. In the implementation of the proposed algorithm, the area is
divided into $1\times1\ \text{m}^{2}$ grid cells such that the $166$
measurements are in the grid center. Among the $166$ measurements,
$M'=20$ \textendash{} 140 measurements are randomly selected for
reconstructing the prorogation field. The results in reconstruction
\ac{mse} are reported in Fig.~\ref{fig:real data}. It is shown
that the proposed NNM-t method also achieves a promising performance
on a real dataset.

\section{Conclusion\label{sec:Conclusion}}

In this paper, an interpolation assisted matrix completion method
is proposed to reconstruct a propagation map from a sparse set of
signal strength measurements. The observation matrix is enriched through
interpolation and the bias and variance of the interpolation method
are analyzed to build uncertainty-aware matrix completion algorithms
which integrate interpolation with matrix completion. Furthermore,
a minimum MSE method to optimize the window size for local interpolation
is proposed. The numerical results demonstrate that the reconstruction
MSE of the propagation map can be reduced by $10$\%\textendash $50$\%
compared to the state-of-the-art schemes which corresponds to a saving
of nearly half of measurements under a not-so-dense sampling. In the
application of RSS-based source localization, the RMSE of localization
is reduced by more than $50$\% from a WCL baseline.


\appendices{}


\section{Proof of Theorem \ref{thm1:0-th bv} \label{app:theorem1}}

Applying first-order Taylor's expansion to $\rho(\bm{z})$ at the
neighborhood of $\bm{c}_{ij}$, the observation model in (\ref{eq:model-measurement})
becomes 
\begin{align}
\gamma_{m} & =\rho(\text{\ensuremath{\bm{z}_{m}}})+\epsilon_{m}\nonumber \\
 & =\rho(\bm{c}_{ij})+\nabla\rho^{\text{T}}(\bm{c}_{ij})(\bm{z}_{m}-\bm{c}_{ij})+o(||\bm{z}_{m}-\bm{c}_{ij}||)+\epsilon_{m}\label{eq:first order expan}
\end{align}
due to the first order differentiability of $\rho(\bm{z})$.

Recall that $\bar{w}_{m}(\bm{c}_{ij})\triangleq w_{m}(\bm{c}_{ij})/\sum_{i=1}^{M}w_{i}(\bm{c}_{ij})$,
and therefore, $\sum_{m=1}^{M}\bar{w}_{m}(\bm{c}_{ij})=1$. The expectation
of the zeroth order solution $\hat{\alpha}(\bm{c}_{ij})$ in (\ref{eq:sol-0th})
can be written as

\begin{align}
 & \mathbb{E}\left\{ \hat{\alpha}(\bm{c}_{ij})\right\} \nonumber \\
 & =\mathbb{E}\left\{ \frac{\sum_{m=1}^{M}w_{m}(\bm{c}_{ij})\gamma_{m}}{\sum_{m=1}^{M}w_{m}(\bm{c}_{ij})}\right\} \nonumber \\
 & =\mathbb{E}\Bigg\{\sum_{m=1}^{M}\bar{w}_{m}(\bm{c}_{ij})(\rho(\bm{c}_{ij})+\nabla\rho^{\text{T}}(\bm{c}_{ij})(\bm{z}_{m}-\bm{c}_{ij})\nonumber \\
 & \qquad\quad+o(||\bm{z}_{m}-\bm{c}_{ij}||)+\epsilon_{m})\Bigg\}\nonumber \\
 & =\mathbb{E}\Bigg\{\sum_{m=1}^{M}\bar{w}_{m}(\bm{c}_{ij})(\rho(\bm{c}_{ij})+\nabla\rho^{\text{T}}(\bm{c}_{ij})(\bm{z}_{m}-\bm{c}_{ij})\nonumber \\
 & \qquad\quad+o(||\bm{z}_{m}-\bm{c}_{ij}||))\Bigg\}+\sum_{m=1}^{M}\bar{w}_{m}(\bm{c}_{ij})\mathbb{E}\left\{ \epsilon_{m}\right\} \nonumber \\
 & =\rho(\bm{c}_{ij})+\sum_{m=1}^{M}\bar{w}_{m}(\bm{c}_{ij})(\nabla\rho^{\text{T}}(\bm{c}_{ij})(\bm{z}_{m}-\bm{c}_{ij}))+o(b)\label{eq:expectation-0th}
\end{align}
where $o(b)$ in the last equation is due to the fact that $o(||\bm{z}_{m}-\bm{c}_{ij}||)/b\to0$
as $b\to0$ since we are only interested in the $(i,j)$th grid such
that $||\bm{z}_{m}-\bm{c}_{ij}||<b$ according to the interpolation
strategy.

As a result, the bias
\[
\begin{aligned}\mathbb{E}\{\xi_{ij}\} & =\mathbb{E}\left\{ \hat{\alpha}(\bm{c}_{ij})\right\} -\rho(\bm{c}_{ij})\\
 & =\sum_{m=1}^{M}\bar{w}_{m}(\bm{c}_{ij})(\nabla\rho^{\text{T}}(\bm{c}_{ij})(\bm{z}_{m}-\bm{c}_{ij}))+o(b).
\end{aligned}
\]

For the variance of $\xi_{ij}=\hat{\alpha}(\bm{c}_{ij})-\rho(\bm{c}_{ij})$,
we note that $\mathbb{V}\left\{ \xi_{ij}\right\} =\mathbb{V}\left\{ \hat{\alpha}(\bm{c}_{ij})\right\} $
because $\rho(\bm{c}_{ij})$ is deterministic. From (\ref{eq:sol-0th})
and (\ref{eq:model-measurement}),

\[
\begin{aligned}\hat{\alpha}(\bm{c}_{ij}) & =\sum_{m=1}^{M}\bar{w}_{m}\rho(\bm{c}_{ij})+\sum_{m=1}^{M}\bar{w}_{m}\epsilon_{m}.\end{aligned}
\]

So, $\mathbb{V}\left\{ \xi_{ij}\right\} =\mathbb{V}\left\{ \sum_{m=1}^{M}\bar{w}_{m}\epsilon_{m}\right\} =\sum_{m=1}^{M}\bar{w}_{m}^{2}(\bm{c}_{ij})\sigma^{2}$,
since $\mathbb{V}\left\{ \epsilon_{m}\right\} =\sigma^{2}$ from the
observation model in (\ref{eq:model-measurement}). The results in
Theorem 1 are thus proven.

\section{Proof of Theorem \ref{thm2:1st bv}\label{app:the2}}

Applying second order Taylor's expansion to $\rho(\bm{z}_{m})$ at
the neighborhood of $\bm{c}_{ij}$, we have 
\begin{equation}
\begin{aligned}\rho\left(\bm{z}_{m}\right) & =\rho(\bm{c}_{ij})+\nabla\rho^{\text{T}}(\bm{c}_{ij})(\bm{z}_{m}-\bm{c}_{ij})\\
 & +\frac{1}{2}(\bm{z}_{m}-\bm{c}_{ij})^{\text{T}}\bm{\Psi}_{ij}(\bm{z}_{m}-\bm{c}_{ij})+o(\|\bm{z}_{m}-\bm{c}_{ij}\|^{2})
\end{aligned}
\label{eq:second order rho}
\end{equation}
where $\bm{\Psi}_{ij}=\nabla^{2}\rho(\bm{c}_{ij})$ is the Hessian
matrix of $\rho(\bm{z})$ evaluated at point $\bm{c}_{ij}$.

Denote $\bm{\varrho}\triangleq\left[\rho(\bm{z}_{1}),\cdots,\rho(\bm{z}_{M})\right]^{\text{T}}$.
Then the expression in (\ref{eq:second order rho}) can be rearranged
into the following matrix form 
\begin{equation}
\bm{\varrho}=\tilde{\bm{D}}_{ij}^{\text{T}}\left[\begin{array}{c}
\rho(\bm{c}_{ij})\\
\nabla\rho(\bm{c}_{ij})
\end{array}\right]+\frac{1}{2}\text{diag}\{\bm{D}_{ij}^{\text{T}}\bm{\Psi}_{ij}\bm{D}_{ij}\}+\bm{r}_{ij}\label{eq:matrixform_firstorder}
\end{equation}
where $\bm{r}_{ij}$ is a vector of the residual terms $o(||\bm{z}_{m}-\bm{c}_{ij}||^{2})$
. Noticing that the residual from the Taylor's expansion scales as
$||\bm{z}_{m}-\bm{c}_{ij}||^{2}$ with a bound $||\bm{z}_{m}-\bm{c}_{ij}||<b$
due to the interpolation strategy, thus we have $o(||\bm{z}_{m}-\bm{c}_{ij}||^{2})/b^{2}\to0$
as $b\to0$. Therefore, $\bm{r}{}_{ij}\sim o(b^{2})$.

From the measurement model (\ref{eq:model-measurement}), we have
$\bm{\gamma}=\bm{\varrho}+\bm{\epsilon}$ , where $\bm{\epsilon}=[\epsilon_{1},\cdots,\epsilon_{M}]^{\text{T}}$.
The expectation of the first order solution $\hat{\bm{\theta}}=\big[\hat{\alpha}(\bm{c}_{ij})\quad\hat{\bm{\beta}}(\bm{c}_{ij})\big]^{\text{T}}$
in (\ref{eq:firstsol}) can be written as 
\begin{equation}
\begin{aligned}\mathbb{E}\{\hat{\bm{\theta}}\} & =\mathbb{E}\left\{ \tilde{\bm{W}}_{ij}^{-1}\tilde{\bm{D}}_{ij}\bm{W}_{ij}\bm{\gamma}\right\} =\tilde{\bm{W}}_{ij}^{-1}\tilde{\bm{D}}_{ij}\bm{W}_{ij}\bm{\varrho}\end{aligned}
\label{eq:epctheta}
\end{equation}
since the noise $\bm{\epsilon}$ has zero mean, where $\tilde{\bm{W}}_{ij}=\tilde{\bm{D}}_{ij}\bm{W}_{ij}\tilde{\bm{D}}_{ij}^{\text{T}}$.

Substitute $\bm{\varrho}$ in (\ref{eq:epctheta}) with (\ref{eq:matrixform_firstorder}),
the expectation of $\hat{\bm{\theta}}$ can be rewritten as 
\[
\begin{aligned} & \mathbb{E}\big\{\hat{\bm{\theta}}\big\}\\
= & \tilde{\bm{W}}_{ij}^{-1}\tilde{\bm{D}}_{ij}\bm{W}_{ij}\\
 & \text{\ensuremath{\quad}}\times\left(\tilde{\bm{D}}_{ij}^{\text{T}}\left[\begin{array}{c}
\rho(\bm{c}_{ij})\\
\nabla\rho(\bm{c}_{ij})
\end{array}\right]+\frac{1}{2}\text{diag}\{\bm{D}_{ij}^{\text{T}}\bm{\Psi}_{ij}\bm{D}_{ij}\}+\bm{r}_{ij}\right)\\
= & \left[\begin{array}{c}
\rho(\bm{c}_{ij})\\
\nabla\rho(\bm{c}_{ij})
\end{array}\right]\\
 & \quad+\tilde{\bm{W}}_{ij}^{-1}\tilde{\bm{D}}_{ij}\bm{W}_{ij}\left(\frac{1}{2}\text{diag}\{\bm{D}_{ij}^{\text{T}}\bm{\Psi}_{ij}\bm{D}_{ij}\}\right)+o(b^{2}).
\end{aligned}
\]
As a result, the bias 
\begin{align*}
\mathbb{E}\{\xi_{ij}\} & =\left[\mathbb{E}\big\{\hat{\bm{\theta}}\big\}-\left[\begin{array}{c}
\rho(\bm{c}_{ij})\\
\nabla\rho(\bm{c}_{ij})
\end{array}\right]\right]_{(1,1)}\\
 & =\frac{1}{2}\left[\tilde{\bm{W}}_{ij}^{-1}\tilde{\bm{D}}_{ij}\bm{W}_{ij}\text{diag}\{\bm{D}_{ij}^{\text{T}}\bm{\Psi}_{ij}\bm{D}_{ij}\}\right]_{(1,1)}+o(b^{2})
\end{align*}
where the operation $[\bm{A}]_{(1,1)}$ returns the $(1,1)$th entry
of a matrix $\bm{A}$.

The variance of $\hat{\bm{\theta}}$ can be derived following 
\[
\begin{aligned}\mathbb{V}\{\hat{\bm{\theta}}\} & =\mathbb{E}\left\{ \left(\hat{\bm{\theta}}-\mathbb{E}\{\hat{\bm{\theta}}\}\right)^{2}\right\} \\
 & =\mathbb{E}\left\{ \left(\tilde{\bm{W}}_{ij}^{-1}\tilde{\bm{D}}_{ij}\bm{W}_{ij}\bm{\gamma}-\tilde{\bm{W}}_{ij}^{-1}\tilde{\bm{D}}_{ij}\bm{W}_{ij}\bm{\varrho}\right)^{2}\right\} \\
 & =\mathbb{E}\left\{ \left(\tilde{\bm{W}}_{ij}^{-1}\tilde{\bm{D}}_{ij}\bm{W}_{ij}(\bm{\gamma}-\bm{\varrho})\right)^{2}\right\} \\
 & =\mathbb{E}\left\{ \left(\tilde{\bm{W}}_{ij}^{-1}\tilde{\bm{D}}_{ij}\bm{W}_{ij}\epsilon\right)^{2}\right\} \\
 & =\sigma^{2}\tilde{\bm{W}}_{ij}^{-1}\left(\tilde{\bm{D}}_{ij}\bm{W}_{ij}\bm{W}_{ij}^{\text{T}}\tilde{\bm{D}}_{ij}^{\text{T}}\right)\tilde{\bm{W}}_{ij}^{-1}.
\end{aligned}
\]

Since $\bm{\theta}$ is deterministic, $\mathbb{V}\{\hat{\bm{\theta}}-\bm{\theta}\}=\mathbb{V}\{\hat{\bm{\theta}}\}$,
therefore $\begin{aligned}\mathbb{V}\{\hat{\bm{\theta}}-\bm{\theta}\} & =\sigma^{2}\tilde{\bm{W}}_{ij}^{-1}\left(\tilde{\bm{D}}_{ij}\bm{W}_{ij}\bm{W}_{ij}^{\text{T}}\tilde{\bm{D}}_{ij}^{\text{T}}\right)\tilde{\bm{W}}_{ij}^{-1}.\end{aligned}
$

As a result, 
\begin{align*}
\mathbb{V}\{\xi_{ij}\} & =\left[\mathbb{V}\{\hat{\bm{\theta}}-\bm{\theta}\}\right]_{(1,1)}\\
 & =\sigma^{2}\left[\tilde{\bm{W}}_{ij}^{-1}\left(\tilde{\bm{D}}_{ij}\bm{W}_{ij}\bm{W}_{ij}^{\text{T}}\tilde{\bm{D}}_{ij}^{\text{T}}\right)\tilde{\bm{W}}_{ij}^{-1}\right]_{(1,1)}.
\end{align*}

\section{Proof of Theorem \ref{thm:asymp 0th}\label{app:the3}}

The observation model in (\ref{eq:model-measurement}) can be rewritten
as

\begin{equation}
\begin{aligned}\gamma_{m} & =\rho(\bm{z}_{m})+\epsilon_{m}\\
 & =\rho(\bm{c}_{ij})+(\rho(\bm{z}_{m})-\rho(\bm{c}_{ij}))+\epsilon_{m}.
\end{aligned}
\label{eq:asymp zeroth order model}
\end{equation}

From (\ref{eq:asymp zeroth order model}), we have
\begin{align}
 & \frac{1}{Mb^{2}}\sum_{m=1}^{M}w_{m}(\bm{c}_{ij})\gamma_{m}\qquad\qquad\qquad\qquad\qquad\qquad\nonumber \\
 & =\frac{1}{Mb^{2}}\sum_{m=1}^{M}w_{m}(\bm{c}_{ij})\rho(\bm{c}_{ij})\nonumber \\
 & \quad\quad+\frac{1}{Mb^{2}}\sum_{m=1}^{M}w_{m}(\bm{c}_{ij})(\rho(\bm{z}_{m})-\rho(\bm{c}_{ij}))\nonumber \\
 & \qquad\quad+\frac{1}{Mb^{2}}\sum_{m=1}^{M}w_{m}(\bm{c}_{ij})\epsilon_{m}\label{eq:three terms}
\end{align}
where in the rest of the proof, the three terms in (\ref{eq:three terms})
are respectively denoted as 
\begin{equation}
\hat{f}(\bm{c}_{ij})\triangleq\frac{1}{Mb^{2}}\sum_{m=1}^{M}w_{m}(\bm{c}_{ij})\label{eq:f hat}
\end{equation}
which serves as an approximation of the sensor density at point $\bm{c}_{ij}$,
\begin{equation}
\varepsilon_{1}(\bm{c}_{ij})\triangleq\frac{1}{Mb^{2}}\sum_{m=1}^{M}w_{m}(\bm{c}_{ij})(\rho(\bm{z}_{m})-\rho(\bm{c}_{ij}))\label{eq:epsilon1}
\end{equation}
which quantifies the weighted average bias of the estimate at location
$\bm{c}_{ij}$, and
\begin{equation}
\varepsilon_{2}(\bm{c}_{ij})\triangleq\frac{1}{Mb^{2}}\sum_{m=1}^{M}=w_{m}(\bm{c}_{ij})\epsilon_{m}\label{eq:epsilon 2}
\end{equation}
which quantifies the weighted average error due to the measurement
noise.

Then, dividing $\hat{f}(\bm{c}_{ij})$ on both sides of (\ref{eq:three terms})
and according to the zeroth order solution (\ref{eq:sol-0th}), it
yields 
\begin{align*}
\hat{\rho}(\bm{c}_{ij}) & =\frac{\sum_{m=1}^{M}w_{m}(\bm{c}_{ij})\gamma_{m}}{\sum_{m=1}^{M}w_{m}(\bm{c}_{ij})}\\
 & =\rho(\bm{c}_{ij})+\frac{\varepsilon_{1}(\bm{c}_{ij})}{\hat{f}(\bm{c}_{ij})}+\frac{\varepsilon_{2}(\bm{c}_{ij})}{\hat{f}(\bm{c}_{ij})}
\end{align*}
which implies that the error equals to
\begin{equation}
\begin{aligned}\xi_{ij} & =\hat{\rho}(\bm{c}_{ij})-\rho(\bm{c}_{ij})=\frac{\varepsilon_{1}(\bm{c}_{ij})}{\hat{f}(\bm{c}_{ij})}+\frac{\varepsilon_{2}(\bm{c}_{ij})}{\hat{f}(\bm{c}_{ij})}.\end{aligned}
\label{eq:zeroth aysmp model}
\end{equation}

To analyze $\varepsilon_{1}(\bm{c}_{ij})$ and $\varepsilon_{2}(\bm{c}_{ij})$,
we derive the following lemmas.
\begin{lem}
\label{lem: f} Suppose that $f(\bm{z})$ is second order differentiable
and bounded. Then, for a sufficiently small $b$,
\[
\int K\Big(\frac{\bm{x}-\bm{c}_{ij}}{b}\Big)^{2}f(\bm{x})d\bm{x}=C_{1}b^{2}f(\bm{c}_{ij})+o\left(b^{3}\right).
\]
In addition, for a function $g(\bm{x})$ that satisfies $g(b\bm{u})/b\to0$
as $b\to0,$ uniformly for all $||\bm{u}||\leq1,$ there is
\[
\frac{1}{b}\int K(\bm{u})g(b\bm{u})f(\bm{z})d\bm{u}\to0
\]
 as $b\to0,$ for all $\bm{z}$, where $C_{1}=\int K(\bm{u})^{2}d\bm{u}$.
\end{lem}
\begin{proof}
Let $\bm{u}=(\bm{x}-\bm{c}_{ij})/b$. Note that $\int d\bm{x}=\int db\bm{u}=\iint dbu_{x}dbu_{y}=\iint b^{2}du_{x}du_{y}=\int b^{2}d\bm{u}$.
We have
\begin{equation}
\int K\Big(\frac{\bm{x}-\bm{c}_{ij}}{b}\Big)^{2}f(\bm{x})d\bm{x}=b^{2}\int K(\bm{u})^{2}f(b\bm{u}+\bm{c}_{ij})d\bm{u}.\label{eq:ku2f}
\end{equation}
Consider the first order Taylor's expansion
\begin{equation}
f(b\bm{u}+\bm{c}_{ij})=f(\bm{c}_{ij})+\nabla f(\bm{c}_{ij})^{\text{T}}b\bm{u}+R_{f}(b\bm{u})\label{eq:f taylor}
\end{equation}
where $\nabla f$ is the first order derivative of $f$, $R_{f}(b\bm{u})$
is the remainder term that satisfies $R_{f}(b\bm{u})/b\to0$ as $b\to0$
uniformly for all $\bm{u}$ with $\|\bm{u}\|\leq1$ due to the Taylor's
theorem and the fact that $f(\bm{z})$ is second order differentiable.
Thus,
\begin{align}
 & \int K\left(\bm{u}\right)^{2}f(b\bm{u}+\bm{c}_{ij})d\bm{u}\nonumber \\
 & =\int K\left(\bm{u}\right)^{2}\left(f(\bm{c}_{ij})+\nabla f(\bm{c}_{ij})^{\text{T}}b\bm{u}+R_{f}(b\bm{u})\right)d\bm{u}\nonumber \\
 & =f(\bm{c}_{ij})\int K\left(\bm{u}\right)^{2}d\bm{u}+b\int K\left(\bm{u}\right)^{2}\nabla f(\bm{c}_{ij})^{\text{T}}\bm{u}d\bm{u}\label{eq:2 term}\\
 & \qquad+\int K\left(\bm{u}\right)^{2}R_{f}(b\bm{u})d\bm{u}\label{eq:o(b)-1}\\
 & =C_{1}f(\bm{c}_{ij})+o\left(b\right)\label{eq:inter result}
\end{align}
where the second term in (\ref{eq:2 term}) equals $0$, since $\iint K\left(\bm{u}\right)^{2}(a_{1}u_{x}+a_{2}u_{y})du_{x}du_{y}=0$
due to (\ref{eq:integral 0}). For (\ref{eq:o(b)-1}), since the support
of $K\left(\bm{u}\right)$ is $\mathcal{C}=\left\{ \text{\ensuremath{\bm{u}}}\in\mathbb{R}^{2}:||\bm{u}||_{2}<1\right\} $,
we have 
\begin{equation}
\begin{alignedat}{1} & \frac{1}{b}\int K\left(\bm{u}\right)^{2}R_{f}(b\bm{u})d\bm{u}\\
\leq\  & \frac{1}{b}\bigg|\underset{\bm{u}\in\mathcal{C}}{\sup}\ R_{f}(b\bm{u})\bigg|\cdot\bigg|\int_{\mathcal{C}}K\left(\bm{u}\right)^{2}d\bm{u}\bigg|
\end{alignedat}
\qquad\qquad\label{eq:o(b)}
\end{equation}
which converges to $0$ because the term $R_{f}(b\bm{u})/b$ uniformly
converges to $0$ and the integral term is bounded.

Multiplying (\ref{eq:inter result}) with $b^{2}$ and substituting
the result in (\ref{eq:ku2f}) confirms the result in Lemma \ref{lem: f}.
\end{proof}
Define a random variable $X_{m}=w_{m}(\bm{c}_{ij})\epsilon_{m}$,
where recall that $w_{m}(\bm{c}_{ij})=K\big(\frac{\bm{z}_{m}-\bm{c}_{ij}}{b}\big)$.
Then, using Lemma \ref{lem: f}, the random variable $X_{m}$ can
be shown to have the following property.
\begin{lem}
\label{lem:Xm}The mean and variance of $X_{m}$ are given by $\mathbb{E}\big\{ X_{m}\big\}=0$
and $\mathbb{V}\big\{ X_{m}\big\}=C_{1}\sigma^{2}b^{2}f(\bm{c}_{ij})+o(b^{3}).$
\end{lem}
\begin{proof}
Since $\mathbb{E}\{\epsilon_{m}\}=0$ and $\epsilon_{m}$ is independent
of the term $w_{m}(\bm{c}_{ij})$, it follows that $\mathbb{E}\big\{ X_{m}\big\}=0.$

Moreover,

\begin{align}
\mathbb{V}\left\{ X_{m}\right\}  & =\mathbb{E}\left\{ K\Big(\frac{\bm{z}_{m}-\bm{c}_{ij}}{b}\Big)^{2}\epsilon_{m}^{2}\right\} \nonumber \\
 & =\mathbb{E}\left\{ K\Big(\frac{\bm{z}_{m}-\bm{c}_{ij}}{b}\Big)^{2}\right\} \mathbb{E}\big\{\epsilon_{m}^{2}\big\}\nonumber \\
 & =\sigma^{2}\mathbb{E}\left\{ K\Big(\frac{\bm{z}_{m}-\bm{c}_{ij}}{b}\Big)^{2}\right\} \nonumber \\
 & =\sigma^{2}\int K\Big(\frac{\bm{z}_{m}-\bm{c}_{ij}}{b}\Big)^{2}f(\bm{x})d\bm{x}\nonumber \\
 & =C_{1}\sigma^{2}b^{2}f(\bm{c}_{ij})+o(b^{3})\label{eq:V Xm}
\end{align}
where the last equation is from Lemma \ref{lem: f}.
\end{proof}
Define a random variable $Y_{m}=w_{m}(\bm{c}_{ij})(\rho(\bm{z}_{m})-\rho(\bm{c}_{ij}))$.
We show the property of $Y_{m}$ as follows.
\begin{lem}
\label{lem:Ym}The random variable $Y_{m}$ has mean $\mathbb{E}\left\{ Y_{m}\right\} =b^{4}C_{0}\big(\vartheta_{1}(\bm{c}_{ij})+\frac{1}{2}f(\bm{c}_{ij})\vartheta_{2}(\bm{c}_{ij})\big)+o(b^{4})$
and its variance satisfies $\mathbb{V}\left\{ Y_{m}\right\} /b^{3}\to0$
as $b\to0,$ where $\vartheta_{1}(\bm{c}_{ij})=\frac{\partial f(\bm{c}_{ij})}{\partial u_{x}}\frac{\partial\rho(\bm{c}_{ij})}{\partial u_{x}}+\frac{\partial f(\bm{c}_{ij})}{\partial u_{y}}\frac{\partial\rho(\bm{c}_{ij})}{\partial u_{y}}$
and $\vartheta_{2}(\bm{c}_{ij})=\frac{\partial^{2}\rho(\bm{c}_{ij})}{\partial u_{x}^{2}}+\frac{\partial^{2}\rho(\bm{c}_{ij})}{\partial u_{y}^{2}}$.
\end{lem}
\begin{proof}
Let $\bm{u}=(\bm{x}-\bm{c}_{ij})/b$. We have 
\begin{equation}
\begin{aligned}\mathbb{E}\left\{ Y_{m}\right\}  & =\int K\Big(\frac{\bm{x}-\bm{c}_{ij}}{b}\Big)(\rho(\bm{x})-\rho(\bm{c}_{ij}))f(\bm{x})d\bm{x}\\
 & =b^{2}\int K(\bm{u})(\rho(\bm{c}_{ij}+b\bm{u})-\rho(\bm{c}_{ij}))f(\bm{c}_{ij}+b\bm{u})d\bm{u}.
\end{aligned}
\label{eq:E rho1}
\end{equation}

Consider the first order Taylor expansion of $f(\bm{c}_{ij}+b\bm{u})$
as (\ref{eq:f taylor}) and the second order Taylor expansion
\[
\rho(\bm{c}_{ij}+b\bm{u})=\rho(\bm{c}_{ij})+\nabla\rho(\bm{c}_{ij})^{\text{T}}b\bm{u}+\frac{1}{2}b^{2}\bm{u}^{\text{T}}\bm{\Psi}_{ij}\bm{u}+R_{\rho}(b\bm{u}).
\]
We have
\begin{align}
 & \int K(\bm{u})(\rho(\bm{c}_{ij}+b\bm{u})-\rho(\bm{c}_{ij}))f(\bm{c}_{ij}+b\bm{u})d\bm{u}\nonumber \\
 & =\int K(\bm{u})\big(\nabla\rho(\bm{c}_{ij})^{\text{T}}b\bm{u}+\frac{1}{2}b^{2}\bm{u}^{\text{T}}\bm{\Psi}_{ij}\bm{u}+R_{\rho}(b\bm{u})\big)\nonumber \\
 & \qquad\times\left(f(\bm{c}_{ij})+\nabla f(\bm{c}_{ij})^{\text{T}}b\bm{u}+R_{f}(b\bm{u})\right)d\bm{u}.\label{eq:expand long}
\end{align}
Recall that $K\left(\bm{u}\right)$ has a bounded support and the
terms $R_{f}(b\bm{u})/b$ and $R_{\rho}(b\bm{u})/b^{2}$ uniformly
converge to $0$ as $b\to0$ due to the Taylor expansion. Then, the
terms involving $R_{f}(b\bm{u})$ and $R_{\rho}(b\bm{u})$ from expanding
the product in (\ref{eq:expand long}) are $o(b^{2})$ due to Lemma
\ref{lem: f}. As a result, equation (\ref{eq:expand long}) simplifies
to
\begin{align}
 & b^{2}\int K(\bm{u})\nabla\rho(\bm{c}_{ij})^{\text{T}}\bm{u}\nabla f(\bm{c}_{ij})^{\text{T}}\bm{u}d\bm{u}\nonumber \\
 & \quad+\frac{b^{2}}{2}\int K(\bm{u})\bm{u}{}^{\text{T}}\bm{\Psi}_{ij}\bm{u}f(\bm{c}_{ij})d\bm{u}\nonumber \\
 & =b^{2}\int K(\bm{u})\bigg(\frac{\partial\rho(\bm{c}_{ij})}{\partial u_{x}}u_{x}+\frac{\partial\rho(\bm{c}_{ij})}{\partial u_{y}}u_{y}\bigg)\nonumber \\
 & \qquad\qquad\times\bigg(\frac{\partial f(\bm{c}_{ij})}{\partial u_{x}}u_{x}+\frac{\partial f(\bm{c}_{ij})}{\partial u_{y}}u_{y}\bigg)du_{x}du_{y}\nonumber \\
 & \quad+\frac{b^{2}}{2}\int K(\bm{u})\bigg(\frac{\partial^{2}\rho(\bm{c}_{ij})}{\partial u_{x}^{2}}u_{x}^{2}+2\frac{\partial^{2}\rho(\bm{c}_{ij})}{\partial u_{x}\partial u_{y}}u_{x}u_{y}\nonumber \\
 & \qquad\qquad+\frac{\partial^{2}\rho(\bm{c}_{ij})}{\partial u_{y}^{2}}u_{y}^{2}\bigg)f(\bm{c}_{ij})du_{x}du_{y}+o(b^{2})\nonumber \\
 & =b^{2}C_{0}\Bigg\{\frac{\partial f(\bm{c}_{ij})}{\partial u_{x}}\frac{\partial\rho(\bm{c}_{ij})}{\partial u_{x}}+\frac{\partial f(\bm{c}_{ij})}{\partial u_{y}}\frac{\partial\rho(\bm{c}_{ij})}{\partial u_{y}}\nonumber \\
 & \qquad\qquad+\frac{1}{2}\frac{\partial^{2}\rho(\bm{c}_{ij})}{\partial u_{x}^{2}}f(\bm{c}_{ij})+\frac{1}{2}\frac{\partial^{2}\rho(\bm{c}_{ij})}{\partial u_{y}^{2}}f(\bm{c}_{ij})\Bigg\}+o(b^{2})\nonumber \\
 & =b^{2}C_{0}\big(\vartheta_{1}(\bm{c}_{ij})+\frac{1}{2}f(\bm{c}_{ij})\vartheta_{2}(\bm{c}_{ij})\big)+o(b^{2}).\label{eq:E ep1sol}
\end{align}

Thus, multiply (\ref{eq:E ep1sol}) with $b^{2}$, we obtain
\begin{equation}
\mathbb{E}\left\{ Y_{m}\right\} =b^{4}C_{0}\big(\vartheta_{1}(\bm{c}_{ij})+\frac{1}{2}f(\bm{c}_{ij})\vartheta_{2}(\bm{c}_{ij})\big)+o(b^{4}).\label{eq:E Ym}
\end{equation}

Similarly,
\begin{equation}
\begin{aligned} & \mathbb{V}\left\{ Y_{m}\right\} \\
 & =\mathbb{E}\left\{ Y_{m}^{2}\right\} -\mathbb{E}\left\{ Y_{m}\right\} ^{2}\\
 & =\int\Big(K\Big(\frac{\bm{x}-\bm{c}_{ij}}{b}\Big)(\rho(\bm{x})-\rho(\bm{c}_{ij}))\Big)^{2}f(\bm{x})d\bm{x}-\mathbb{E}\left\{ Y_{m}\right\} ^{2}\\
 & =b^{2}\int\big(K\left(\bm{u}\right)(\nabla\rho(\bm{c}_{ij})^{\text{T}}b\bm{u}+\frac{1}{2}b^{2}\bm{u}^{\text{T}}\bm{\Psi}_{ij}\bm{u}+R_{\rho}(b\bm{u}))\big)^{2}\\
 & \qquad\quad\times f(\bm{c}_{ij}+b\bm{u})d\bm{u}-\mathbb{E}\left\{ Y_{m}\right\} ^{2}\\
 & =b^{4}\int\big(K\left(\bm{u}\right)(\nabla\rho(\bm{c}_{ij})^{\text{T}}\bm{u}+\frac{1}{2}b\bm{u}^{\text{T}}\bm{\Psi}_{ij}\bm{u}+R_{\rho}(b\bm{u}))/b\big)^{2}\\
 & \qquad\quad\times f(\bm{c}_{ij}+b\bm{u})d\bm{u}-\mathbb{E}\left\{ Y_{m}\right\} ^{2}
\end{aligned}
\label{eq:v delta}
\end{equation}
and thus, $\mathbb{V}\left\{ Y_{m}\right\} /b^{3}\to0$, as $b\to0.$
\end{proof}
Next, it is observed that $\text{Cov}(X_{m},Y_{m})=0$. This is because
that $\mathbb{E}\{X_{m}Y_{m}\}=\mathbb{E}\{w_{m}^{2}(\bm{c}_{ij})(\rho(\bm{z}_{m})-\rho(\bm{c}_{ij}))\}\mathbb{E}\{\text{\ensuremath{\epsilon_{m}}}\}=0$
due to the independence between the zero mean noise $\epsilon_{m}$
and $w_{m}^{2}(\bm{c}_{ij})(\rho(\bm{z}_{m})-\rho(\bm{c}_{ij}))$.
As a result, $\text{Cov}(X_{m},Y_{m})=\mathbb{E}\{X_{m}Y_{m}\}-\mathbb{E}\{X_{m}\}\mathbb{E}\{Y_{m}\}=0$
due to $\mathbb{E}\{X_{m}\}=0$ from Lemma \ref{lem:Xm}. Therefore,
$\mathbb{V}\{X_{m}+Y_{m}\}=\mathbb{V}\{X_{m}\}+\mathbb{V}\{Y_{m}\}$.

In addition, from (\ref{eq:epsilon1}) and (\ref{eq:epsilon 2}),
it is observed that $\varepsilon_{1}(\bm{c}_{ij})+\varepsilon_{2}(\bm{c}_{ij})=\frac{1}{b^{2}}\frac{1}{M}\sum_{m=1}^{M}(X_{m}+Y_{m})$
where the variable $Z_{m}\triangleq X_{m}+Y_{m}$ is i.i.d., since
both $\bm{z}_{m}$ and $\epsilon_{m}$ are i.i.d.. As a result, the
law of large number implies $b^{2}(\varepsilon_{1}(\bm{c}_{ij})+\varepsilon_{2}(\bm{c}_{ij}))\overset{p}{\to}\mathbb{E}\{Z_{m}\}$
as $M\to\infty.$

In addition, the central limit theorem yields 
\begin{align*}
\sqrt{M}\Big(b^{2}(\varepsilon_{1}(\bm{c}_{ij})+\varepsilon_{2}(\bm{c}_{ij}))-\mathbb{E}\{Z_{m}\}\Big) & \overset{d}{\rightarrow}\mathcal{N}(0,\mathbb{V}\{Z_{m}\})
\end{align*}
as $M\to\infty$, where $\mathbb{E}\{Z_{m}\}=\mathbb{E}\{X_{m}\}+\mathbb{E}\{Y_{m}\}$
and $\mathbb{V}\{Z_{m}\}=\mathbb{V}\{X_{m}+Y_{m}\}=\mathbb{V}\{X_{m}\}+\mathbb{V}\{Y_{m}\}$
as shown before.

Finally, it is known that $\hat{f}(\bm{z})\overset{p}{\to}f(\bm{z})$
for every $\bm{z}$ \cite{Parzen:J62}. Then, using Slutsky's theorem
\cite{dasgupta2008asymptotic}, the mean 
\begin{equation}
\frac{b^{2}(\varepsilon_{1}(\bm{c}_{ij})+\varepsilon_{2}(\bm{c}_{ij}))}{\hat{f}(\bm{c}_{ij})}\overset{p}{\to}\frac{\mathbb{E}\{Z_{m}\}}{f(\bm{c}_{ij})}\label{eq:ep1+ep2}
\end{equation}
 for every $\bm{c}_{ij}$. Then, substituting $(\varepsilon_{1}(\bm{c}_{ij})+\varepsilon_{2}(\bm{c}_{ij}))/\hat{f}(\bm{c}_{ij})$
with $\xi_{ij}$ as in (\ref{eq:zeroth aysmp model}) and $\mathbb{E}\{Z_{m}\}$
with the mean of $X_{m}$ and $Y_{m}$ in Lemmas \ref{lem:Xm} and
\ref{lem:Ym}, the result (\ref{eq:asym 0 mean}) in Theorem \ref{thm:asymp 0th}
is obtained.

In addition, defining $\bar{\xi}_{ij}=\xi_{ij}-\mathbb{E}\{\xi_{ij}\}$,
it follows that 
\begin{equation}
\sqrt{Mb^{2}}\bar{\xi}_{ij}\stackrel{d}{\longrightarrow}\frac{\mathcal{N}\Big(0,\mathbb{V}\{Z_{m}\}\Big)}{f(\bm{c}_{ij})}\label{eq:variance xi}
\end{equation}
for every $\bm{c}_{ij}.$ Substituting the variance of $X_{m}$ and
$Y_{m}$ in Lemmas \ref{lem:Xm} and \ref{lem:Ym} to $\mathbb{V}\{Z_{m}\}$
in (\ref{eq:variance xi}), the result (\ref{eq:asym 0 distri}) in
Theorem~\ref{thm:asymp 0th} is obtained.

\bibliographystyle{IEEEtran}
\bibliography{IEEEabrv,StringDefinitions,JCgroup,ChenBibCV}
\begin{IEEEbiography}[{\includegraphics[width=1in,height=1.25in,clip,keepaspectratio]{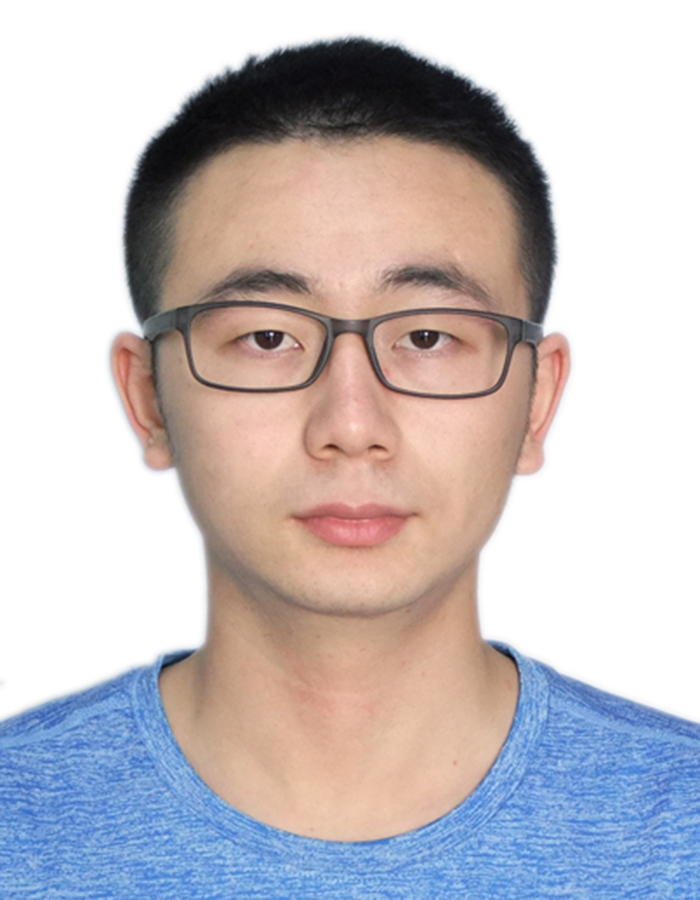}}]{Hao Sun} received the B.S.\ degree in biomedical engineering from the University of Electronic Science and Technology of China (UESTC), Chengdu, China, in 2018. He received two national scholarships during his undergraduate study. He is currently working toward the Ph.D.\ degree with the School of Science and Engineering and the Future Network of Intelligence Institute (FNii) at the Chinese University of Hong Kong, Shenzhen (CUHK-Shenzhen), Guangdong, China. He works on matrix completion and tensor decomposition with application to radio map reconstruction and source localization.

\end{IEEEbiography} 
\begin{IEEEbiography}[{\includegraphics[width=1in,height=1.25in,clip,keepaspectratio]{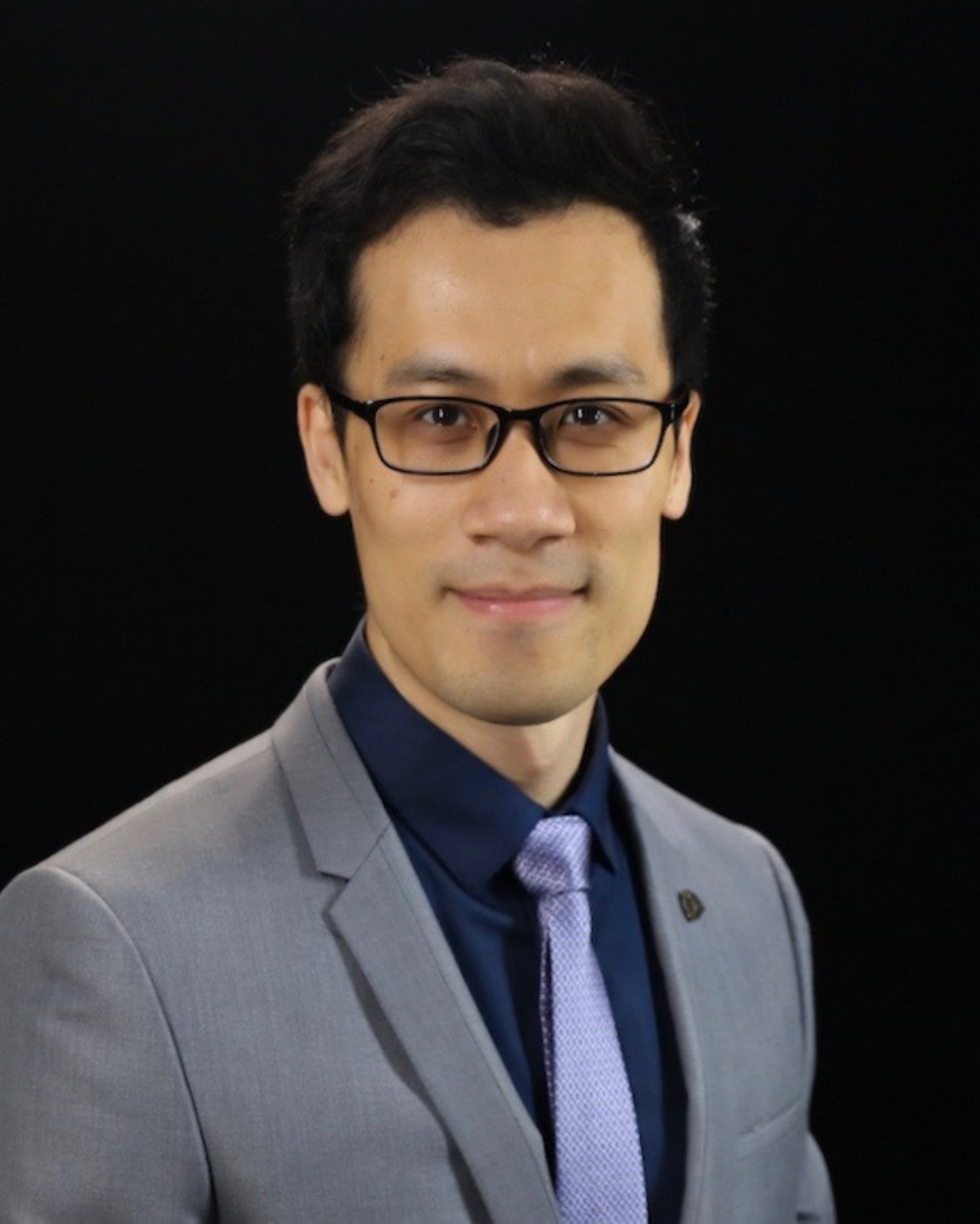}}]{Junting Chen} (S'11--M'16) received the Ph.D.\ degree in Electronic and Computer Engineering from the Hong Kong University of Science and Technology (HKUST), Hong Kong SAR China, in 2015, and the B.Sc.\ degree in Electronic Engineering from Nanjing University, Nanjing, China, in 2009. From 2014--2015, he was a visiting student with the Wireless Information and Network Sciences Laboratory at MIT, Cambridge, MA, USA.  

He is an Assistant Professor with the School of Science and Engineering and the Future Network of Intelligence Institute (FNii) at the Chinese University of Hong Kong, Shenzhen (CUHK-Shenzhen), Guangdong, China. Prior to joining CUHK-Shenzhen, he was a Postdoctoral Research Associate with the Ming Hsieh Department of Electrical Engineering, University of Southern California (USC), Los Angeles, CA, USA, from 2016--2018, and with the Communication Systems Department of EURECOM, Sophia-Antipolis, France, from 2015--2016. He was a recipient of the HKTIIT Post-Graduate Excellence Scholarships in 2012 from HKUST. He works on unimodal signal processing, radio map sensing, UAV assisted communications, and, more generally, machine learning and optimization for wireless communications and localization.

\end{IEEEbiography} 

\end{document}